\documentclass[11pt,letterpaper]{article}

\usepackage{booktabs} 

\newcommand\spacingset[1]{\renewcommand{\baselinestretch}%
  {#1}\small\normalsize}

\usepackage{latexsym}
\usepackage{amssymb,amsmath, bm,pgfplots,tikz,bbm}
\usepackage{graphicx}
\usepackage{marvosym}
\usepackage{multirow,float}
\usepackage{algpseudocode}
\usepackage{caption}
\usepackage{mathtools}
\usepackage{subcaption}
\usepackage{comment}
\usepackage{enumitem}
\allowdisplaybreaks

\usepackage{natbib}
\usepackage{framed}   % grey terminal-style prompt boxes
\usepackage{soul}     % color highlighting in the pipeline figure
\usepackage{longtable}  % must come BEFORE arydshln

\usepackage[colorlinks=true, citecolor=blue, linkcolor=black,
urlcolor=blue, bookmarksopen=true, bookmarksnumbered=true,
pdfstartview=FitH, breaklinks=true]{hyperref}
\usepackage{etoolbox}
\makeatletter
\patchcmd{\NAT@citex}
  {\@citea\NAT@hyper@{\NAT@nmfmt{\NAT@nm}\hyper@natlinkbreak{\NAT@aysep\NAT@spacechar}{\@citeb\@extra@b@citeb}\NAT@date}}
  {\@citea\NAT@nmfmt{\NAT@nm}\NAT@aysep\NAT@spacechar\NAT@hyper@{\NAT@date}}{}{}
\patchcmd{\NAT@citex}
  {\@citea\NAT@hyper@{\NAT@nmfmt{\NAT@nm}\hyper@natlinkbreak{\NAT@spacechar\NAT@@open\if*#1*\else#1\NAT@spacechar\fi}{\@citeb\@extra@b@citeb}\NAT@date}}
  {\@citea\NAT@nmfmt{\NAT@nm}\NAT@spacechar\NAT@@open\if*#1*\else#1\NAT@spacechar\fi\NAT@hyper@{\NAT@date}}{}{}
\makeatother

\usepackage{dcolumn}
\newcolumntype{.}{D{.}{.}{-1}}
\newcolumntype{d}[1]{D{.}{.}{#1}}

\usepackage{theorem}
\theoremstyle{plain}
\theoremheaderfont{\scshape}
\newtheorem{theorem}{Theorem}
\newtheorem{proposition}{Proposition}
\newtheorem{assumption}{Assumption}

\newcommand{\qed}{\hfill \ensuremath{\Box}}
\newcommand{\indep}{\mbox{$\perp\!\!\!\perp$}}

\DeclareMathOperator*{\argmin}{argmin}

\newenvironment{proof}{\vspace{1ex}\noindent{\bf Proof}\hspace{0.5em}}
{\hfill\qed\vspace{1ex}}
\usepackage{kantlipsum}
\allowdisplaybreaks

\graphicspath{{Manuscripts_application/figs/}}

\usepackage{rotating}

\usepackage{arydshln}
\usepackage{threeparttable}

\usepackage[compact]{titlesec}

\usepackage[ruled,linesnumbered,vlined]{algorithm2e}
\usepackage{pifont}

\newtheorem{definition}{Definition}

\newcommand\E{\mathbb{E}}
\newcommand\R{\mathbb{R}}
\newcommand\cR{\mathcal{R}}
\renewcommand\P{\mathbb{P}}

\newcommand\cQ{\mathcal{Q}}

\newcommand\bP{\bm{P}}

\newcommand\bR{\bm{R}}

\newcommand\bh{\bm{h}}

\newcommand\bx{\bm{x}}
\newcommand\bX{\bm{X}}

\newcommand\bW{\bm{W}}

\newcommand\boldf{\bm{f}}
\newcommand\bg{\bm{g}}

\newcommand\bU{\bm{U}}
\newcommand\cD{\mathcal{D}}

\newcommand\cS{\mathcal{S}}

\newcommand\cU{\mathcal{U}}
\newcommand\cF{\mathcal{F}}
\newcommand\cX{\mathcal{X}}

\newcommand\bgamma{\boldsymbol{\gamma}}
\newcommand\btheta{\boldsymbol{\theta}}

\newcommand\blambda{\boldsymbol{\lambda}}
\usetikzlibrary{decorations.markings}
\usetikzlibrary{decorations.pathmorphing}
\usetikzlibrary{shapes.geometric, arrows}
\usetikzlibrary{arrows,decorations.pathmorphing,backgrounds,positioning,fit,matrix}
\usetikzlibrary{shapes,decorations,arrows,calc,arrows.meta,fit,positioning}
\tikzset{auto,node distance =1 cm and 1 cm,semithick,
	state/.style ={circle, draw, minimum width = 0.7 cm},
	point/.style = {circle, draw, inner sep=0.04cm,fill,node contents={}},
	bidirected/.style={Latex-Latex,dashed},
	el/.style = {inner sep=2pt, align=left, sloped}
}

\graphicspath{{figs/}}

\begin{document}

\title{\bf Causal Inference with Video Features as Treatments}

\author{Kentaro Nakamura\thanks{Ph.D. Student, John F. Kennedy School
    of Government, Harvard University. Email:
    \href{mailto:knakamura@g.harvard.edu}{knakamura@g.harvard.edu}.} \and
    Adam Breuer\thanks{Assistant Professor, Department of Computer Science
    and Department of Government, Dartmouth College.
    Email: \href{mailto:adam.breuer@dartmouth.edu}{adam.breuer@dartmouth.edu},
    URL: \href{https://www.adambreuer.com}{https://www.adambreuer.com}.} \and
    Michael H. Crespin\thanks{Professor, Department of Political Science,
    and Director, Carl Albert Congressional Research and Studies Center,
    University of Oklahoma.
    Email: \href{mailto:crespin@ou.edu}{crespin@ou.edu}.} \and
    Bryce J. Dietrich\thanks{Associate Professor, Department of Political
    Science, Purdue University.
    Email: \href{mailto:bdietric@purdue.edu}{bdietric@purdue.edu}.} \and
    Kosuke Imai\thanks{Corresponding author. Edith and Benjamin Geisinger
    Professor, Department of Government and Department of Statistics,
    Harvard University, Cambridge, MA 02138. Phone: 617--384--6778, Email:
    \href{mailto:Imai@Harvard.Edu}{Imai@Harvard.Edu}, URL:
    \href{https://imai.fas.harvard.edu}{https://imai.fas.harvard.edu}.} }
\date{\today}
\maketitle\thispagestyle{empty}
%\tableofcontents\thispagestyle{empty}

\begin{abstract}
We develop the first statistical methodology for causal inference with video features as treatments. Video is the most engaging content modality on the internet. A central causal question is how audience reactions change in response to treatment features that unfold over the course of a video.
Unfortunately, standard causal inference methods are not applicable
because confounding features are latent, high-dimensional, and
dynamically related to both the treatment sequence and the outcome
trajectory. To address these challenges, we first reproduce each video
using a deep generative model and leverage the model's internal representations as learned, low-dimensional summaries of video content for causal estimation. 
We then establish that the average potential-outcome trajectory under dynamic stochastic interventions is nonparametrically identified. 
Lastly, we propose a consistent and asymptotically normal estimator based on a longitudinal neural network architecture. %and show how AI-bot responses can be incorporated to improve statistical efficiency.
We empirically validate our approach by constructing a new causal inference benchmark consisting of $10{,}000$ Super Mario Bros.\texttrademark{} levels played by fixed Mario AI agents, where ground-truth causal effects are known by construction. Finally, we apply our method to television advertisements from the 2020 U.S. presidential campaign and find that increasing the probability of a candidate appearing over time leads to higher average viewer evaluations.
With the proposed methodology, researchers can ask which visual features, appearing at which points in a video, influence audience responses, while benchmarking new methods against datasets with known ground-truth causal effects.
\end{abstract}


\newpage
\section{Introduction}
Video has emerged as a dominant medium of modern communication. Social media platforms such as Instagram and TikTok facilitate the widespread creation and dissemination of short-form video content, enabling billions of users to exchange information, ideas, and experiences. News organizations increasingly incorporate video content into their online reporting, while political campaigns use televised and online advertisements to reach voters. Reflecting the growing influence of video on information exposure and persuasion, communication researchers, political scientists, and other social scientists have increasingly sought to leverage video data to better understand human behavior \citep[e.g.,][]{boussalis2021gender, chan2024neural, chiossi2023short, coppock2020small, kalla2022outside, luo2024choosing, rittmann2025public, roozenbeek2022psychological, teixeira2012emotion, wittenberg2021minimal}.

Despite growing interest in video data, its potential remains underutilized. Existing studies typically estimate the effects of exposure to entire videos, overlooking the fact that videos consist of many features that vary across frames. Because social scientists are often interested in the effects of specific features, such as rhetorical elements, visual cues, or particular images, comparisons at the level of the whole video are often too substantively coarse and statistically inefficient. Recent advances in survey methodology have also made it possible to record respondents' real-time reactions as a video unfolds, generating longitudinal outcome data that capture how responses evolve during exposure \citep[e.g.,][]{maier2007reliability, maier2016breaking, waldvogel2020measuring, boussalis2021gender,ettensperger2023convince}. These data can improve statistical efficiency and enable richer counterfactual analyses of how the timing, duration, and sequencing of video features shape responses.

These settings offer substantial opportunities but also present several methodological challenges. First, video data are extremely high-dimensional. Even a short video contains many frames, each of which may encode visual, textual, auditory, and temporal information. As a result, the raw feature space is far too large to model directly given the sample sizes typically available in social science applications. A common dimension-reduction strategy for unstructured data, especially text, is to use low-dimensional representations such as pre-trained embeddings \citep[e.g.,][]{mikolov2013distributed, devlin_bert_2019, lewis2020bart}. For video data, however, such embeddings remain less developed, and it is often unclear how much substantively relevant information they discard.

Second, causal questions about specific video features are much more difficult to answer than questions about exposure to entire videos. Even when videos are randomly assigned to respondents, the features within a video are not themselves randomized. A treatment feature of interest may therefore be systematically correlated with other features that also affect viewers’ responses \citep{fong_causal_2023}. This problem is especially acute in dynamic settings, where video features often exhibit strong temporal dependence and earlier features can be highly predictive of later ones. Consequently, some counterfactual interventions may be implausible, such as changing a feature at one moment while holding fixed an otherwise unrealistic sequence of surrounding features.

Finally, real-time responses are inherently dynamic. A response recorded at a given moment rarely reflects only the video segment shown at that instant. Rather, viewers update their responses as the video unfolds, so a current response often reflects earlier cues, along with the inertia of prior ratings. Existing studies that use real-time response measures exploit this temporal granularity by aligning response trajectories with televised content and, in some cases, aggregating responses over predefined debate moments or speech phases \citep[e.g.,][]{maier2007reliability, maier2016breaking, waldvogel2020measuring, boussalis2021gender, ettensperger2023convince}. They typically do not, however, model the resulting sequential dependence and thus ignore the carryover effects of video frames.

In this paper, we introduce a methodology for causal inference with video data and real-time outcome measures using generative artificial intelligence (GenAI). Building on the recently proposed GenAI-powered inference (GPI) framework \citep{imai2026causal,imai2026leveraging,nakamura2026genai}, we use GenAI models to reconstruct entire videos and then leverage their internal representations for downstream statistical inference.

Recent advances in GenAI have enabled the generation of highly realistic video content. The internal representations used by these models to generate videos contain the information necessary to reproduce the underlying visual content while providing a substantially lower-dimensional representation that is more amenable to statistical analysis. Because these representations are optimized for video generation, they capture rich features of the original video.

A key advantage of this GPI approach is that it can be applied to existing real-world videos. Modern GenAI models can reconstruct such videos with remarkable fidelity, often producing outputs that are difficult for human observers to distinguish from the originals. This allows researchers to transform complex, high-dimensional video data into informative latent representations that can be used for causal inference and other downstream statistical tasks.

We use the extracted internal representations to identify causal effects for video features that unfold over time. Our framework is based on marginal structural models \citep{robi:hern:brum:00, blackwell_how_2018}, a causal inference framework for studying longitudinal data. The proposed methodology allows a respondent's response at any moment to depend on the prior sequence of videos and their own earlier responses. 

We define our causal estimands using dynamic stochastic interventions. Rather than considering a fixed treatment sequence, we design our dynamic intervention by changing treatment probabilities only in ways that remain consistent with the temporal patterns observed in actual videos. This helps avoid severe overlap violations that arise in high-dimensional longitudinal settings \citep[e.g.,][]{muno:vand:12,  rotnitzky2017multiply,  kennedy_nonparametric_2019, papa:etal:22}. 

In this setting, a key identification challenge is that the treatment sequence may be confounded by unobserved video characteristics, including tone, background imagery, speech content, audio features, and preceding visual context. We show that these confounding factors can be recovered from the internal representations of a GenAI model. Adjusting for the resulting deconfounder, which is a low-dimensional latent representation, enables identification of the entire counterfactual trajectory of real-time responses.

Based on this identification result, we develop an estimation and statistical inference procedure. We first learn the deconfounder using a longitudinal neural network architecture, which further reduces the dimensionality of the GenAI internal representations while retaining the information necessary for causal identification. We then use this learned deconfounder to estimate the target causal quantities of interest. 

Our framework builds on semiparametric theory and accommodates flexible machine learning methods for estimation, avoiding restrictive parametric assumptions about the data-generating processes governing video features and outcome responses. We also develop uniform confidence bands for the entire trajectory of counterfactual real-time responses, enabling researchers to quantify uncertainty and assess how causal effects evolve as a video unfolds.

We empirically validate our methodology on a novel benchmark dataset of 10,000 {\it Super Mario Bros.}\texttrademark{} gameplay videos, in which Mario is controlled by a fixed AI agent \citep{karakovskiy2012mario}. Our setup mirrors a campaign advertising application, where a candidate appears onscreen at various points, and the goal is to estimate how their presence influences viewers' responses. In our Mario environment, the treatment is defined as the appearance of Princess Peach, while the appearance of pipe obstacles serves as a confounder, which is by design strongly correlated with Princess Peach’s appearance. Our objective is to estimate the causal effect of Princess Peach's appearance on the number of jumps taken by the Mario agent. 

By construction, the true treatment effect is zero: although the treatment assignment is strongly confounded, Princess Peach is purely cosmetic and has no effect on gameplay. The Mario AI agent does not register her appearance, so even though the outcome (the Mario agent’s jumps) is generated by a highly complex policy, the ground-truth causal effect of Princess Peach’s appearance is exactly zero.
In this setting, our method accurately recovers the true null effect, whereas naive approaches spuriously attribute a positive effect to Princess Peach's appearance. Even adjusting for coarse pipe counts only partially mitigates this bias, since pipe counts do not capture the spatial layout and temporal evolution of pipe obstacles throughout the video.

% video application paragraph
Finally, we apply the proposed methodology to 849 television advertisements from the 2020 U.S. presidential campaign. We randomly assigned advertisements to respondents, who independently evaluated each advertisement in real time using a feeling thermometer ranging from 0 to 100. Using these continuous ratings, we estimate how the sequence of candidate appearances over the course of an advertisement affects viewers’ evaluations of the advertisement.

Because the timing of candidate appearances is likely to be strategically determined, credible causal inference requires adjusting for confounding visual and textual features that evolve throughout the video. Applying our methodology, we find that increasing the likelihood of candidate appearance over time leads to higher average potential ratings of the advertisement. More broadly, this application demonstrates how GenAI–powered dynamic causal inference can be used to study temporally structured interventions.

\paragraph{Related Literature.}
This paper contributes to the emerging methodological literature on causal and statistical inference with unstructured data, including text \citep[e.g.,][]{egami_how_2022, fong_discovery_2016, fong_causal_2023, gui_causal_2023, pryzant_causal_2021, mozer_matching_2020, roberts_adjusting_2020}, images \citep[e.g.,][]{jerzak_image-based_2023, torres_framework_2024}, and audio \citep[e.g.,][]{knox_dynamic_2021}. Despite these advances, video remains relatively underexplored as an object of causal inference because it combines high-dimensional, multimodal, and temporally dependent information. We address this challenge by extending the GPI framework \citep{imai2026causal,imai2026leveraging, nakamura2026genai} from static unstructured objects to dynamic video treatments with repeated real-time outcomes.

Our work also builds on the literature on causal inference for longitudinal data \citep[e.g.,][]{robi:94, robi:hern:brum:00, hernan2002estimating, blac:13, blackwell_how_2018}. Existing approaches to time-varying treatments typically rely on the assumption of sequential ignorability, which requires treatment assignment at each time point to be unconfounded conditional on the observed treatment, covariate, and outcome histories. In many social science applications, however, this assumption is difficult to justify because important confounders are often unobserved. Our setting differs from standard longitudinal observational studies in a fundamental way: the complete video sequence is observed. This enables us to leverage information contained in the video history to recover latent confounding features and adjust for their influence. By doing so, we relax the reliance on observed histories alone and broaden the scope of credible causal inference for dynamic, unstructured interventions.

The remainder of this paper is organized as follows. Section~\ref{sec::example} introduces our empirical application to 2020 U.S. presidential campaign advertisements. Section~\ref{sec:methods} formalizes the proposed methodological framework for video-as-treatment. Section~\ref{sec:mario} then validates the proposed methodology using a Super Mario Bros.\texttrademark{} environment. Section~\ref{sec:application} applies the method to the presidential campaign advertisement data. Finally, Section~\ref{sec:conclusion} concludes by summarizing the advantages and limitations of the proposed framework and outlining directions for future research.

\section{Empirical Application: Presidential Campaign Advertisement}\label{sec::example}

In this section, we briefly review the literature on campaign advertising and then describe how we measure the treatment (candidate appearance) and outcome (real-time evaluation) variables.  %We consider two versions of the outcome variable: one derived from human responses and the other from AI bot responses.

\subsection{Background}

Since the first political campaign commercials appeared on television in the 1950s, candidates have sought to influence election outcomes through paid media. During the 2024 election cycle, federal and gubernatorial campaigns spent more than \$5 billion on television and radio advertising \citep{ridout2025understanding}, along with an additional \$1.3 billion on digital ads across Meta and Google platforms alone \citep{fowler2025election}. These figures likely understate total political advertising expenditures, as they exclude spending by third-party groups and advertising on many other digital platforms \citep{sheingate2022digital}. Campaigns increasingly rely on digital advertising to target voters in narrowly defined geographic areas \citep{fowler2025election} and to persuade specific demographic groups under favorable conditions \citep{tappin2023quantifying}. Such microtargeted advertising has been shown to affect voter turnout among certain segments of the electorate \citep{aggarwal20232}.

We focus on the causal effects of candidate appearances on voter evaluation. Campaigns strategically decide when and how the candidate should appear in ads. \cite{banda2021candidate} analyze candidate-sponsored campaign advertising data collected by the Wisconsin Advertising Project from 2000--2006.  The authors find that candidates appear in a little more than half of presidential ads and approximately 45\% of congressional and gubernatorial campaign commercials. They argue that candidates show up in ads when their images will be most helpful. For example, %candidates are more likely to be seen in primary rather than general election commercials when party cannot be used to differentiate between opponents. Challengers and open seat candidates appear in ads more than incumbents since these candidates need to introduce themselves to voters. 
candidates are also less likely to appear in contrast and attack ads than they are in promotion ads to try and deflect any potential backlash. \cite{santia2025harnessing} also find candidates are more likely to appear in ads with positive emotions and less likely when negative emotions are expressed.

%There is also a gendered component to this topic with female candidates less likely to appear in presidential and gubernatorial contests but more likely in races for the House of Representatives \citep{banda2021candidate} while Latina candidates are more likely to appear when an issue is related to their ethnicity or gender \citep{santia2023intersection}. Although not limited to candidate appearance, \cite{erfort2026gendered} show that campaigns are strategic when deciding who should appear in an ad, finding commercials are more likely to feature women when ads are targeted towards female audiences on Facebook and Instagram. 

In this paper, we estimate the causal effects of candidate appearance within a campaign video on voter evaluation while adjusting for confounders such as content, tone, and other aspects of political ads.
We next describe our experimental design and data collection procedure. 

\subsection{Data and Measurement}
\label{subsec:experiment}

%For each video, two undergraduate students were asked to separately watch it and they evaluate the video using the feeling thermometer ranging from 0 to 100. Following our motivating example, our treatment of interest is the candidate appearance at each time point and the outcome is the average evaluation of the advertisement by the students.

%\paragraph{Video data.}
We use a dataset of 1,150 television advertisement videos aired during the 2020 U.S. presidential campaign from Wesleyan Media Project (WMP) \citep{fowler_political_2021}. Our analysis focuses on a subset of 849 ads that feature the favored candidate, as identified in the WMP codebook.\footnote{This question asks ``Is the favored candidate shown / mentioned in the ad?'' We only use ads where ``Favored candidate is ascertainable''. In the 2020 WMP Presidential Advertisements codebook (version 1), this would be \texttt{cand} equal to one.} The selected ads range in length from approximately 10 to 120 seconds, with a mean duration of 36.9 seconds, a median duration of 30 seconds, and a standard deviation of 14.4 seconds. The sample includes ads from major candidates such as Joe Biden, who accounts for 389 videos or 45.8\% of the sample, and Donald Trump, who accounts for 128 videos or 15.1\%.

Our unit of analysis is the video segment, with each segment lasting 5 seconds. We exclude the first and last segment from the analysis as the sponsor identification task might be unreliable for those segments. This yields a total of 4,391 video segments. Although other segment lengths are possible, we chose 5 seconds because political campaign advertisements typically do not contain rapid scene changes. Using shorter segments would increase temporal dependence across observations and computational cost, whereas longer segments could encompass multiple events, including the candidate's appearance and disappearance, within a single segment, making it difficult to isolate causal effects. 

%\paragraph{Measuring candidate appearance.}

\begin{figure}[t!]
  \centering
   \begin{subfigure}[t]{0.535\textwidth}
    \centering
    \includegraphics[width=\linewidth]{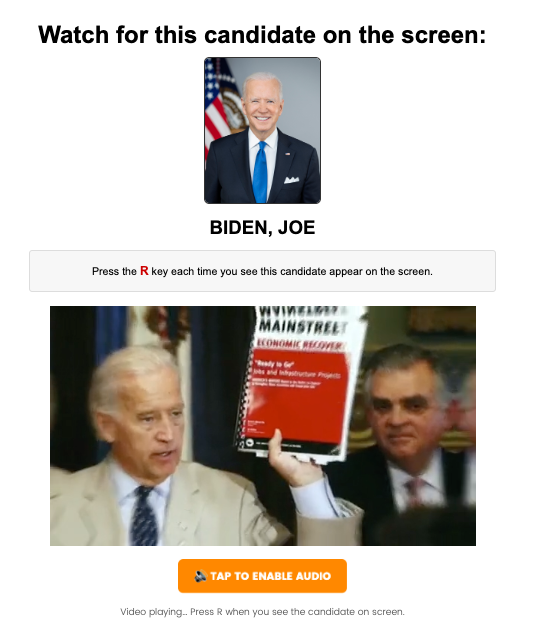}
    \caption{Candidate identification task}
    \label{fig:prolific1b}
  \end{subfigure}\hfill
  \begin{subfigure}[t]{0.465\textwidth}
    \includegraphics[width=\linewidth]{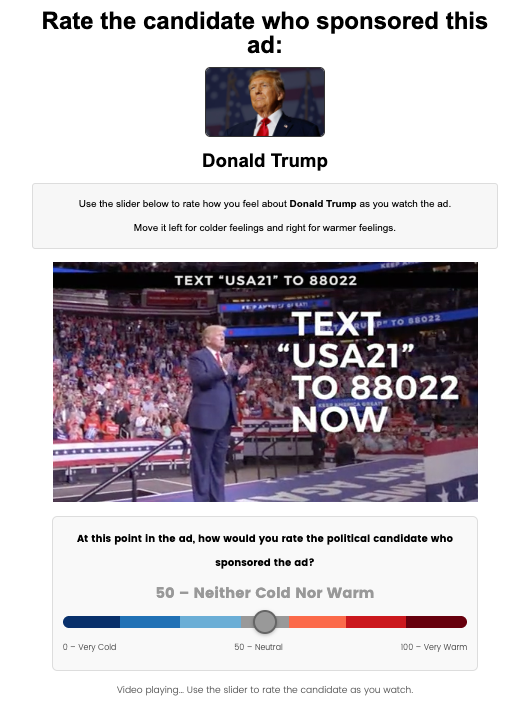}
    \caption{Rating task}
    \label{fig:prolific2b}
  \end{subfigure}
  \caption{Candidate Identification and Rating Tasks.  Panel~\ref{fig:prolific1b} shows the screenshot of the question used for the candidate identification task.  Respondents were asked to press the ``R'' key on their keyboard when the candidate appeared on the screen. An equivalent mobile version of the task was also provided. Panel~\ref{fig:prolific2b} shows the rating task, in which respondents were asked to rate each advertisement continuously and in real time on a 0--100 feeling thermometer by moving the dial as the ad played.}
  \label{fig:surveys}
\end{figure}

Our treatment of interest is the candidate appearance within each
video segment. We fielded a survey using the Prolific platform on June
22, 2026.  Our Prolific sample was restricted to respondents over the
age of 18 who were currently residing in the United States (U.S.) and
were fluent in English.  As shown in Figure~\ref{fig:prolific1b}, each
respondent was asked to complete a simple task of pressing the ``R''
key on their keyboard each time the favored candidate appeared on
screen. Prior to the main task, each respondent was randomly assigned
one of two training tasks. Each ad is assigned to at most five coders.
Appendix~\ref{app:surveys} provides additional details of this survey.
Using the results of this survey, we create the candidate appearance
indicator variable for each ad segment via a simple majority rule
among all coders. In total, 50.5\% of all video segments featured the
video's favored candidate, with Joe Biden appearing most (389 videos,
1282 segments).

We conducted a separate Prolific survey on June 26, 2026, in which
respondents evaluated the favored candidate in real time using a
feeling thermometer. As before, our sample was restricted to
individuals who were at least 18 years old, currently resided in the
U.S., and were fluent in English. To ensure a balanced sample, we used
stratified sampling by partisanship (Democrat, Republican, and
Independent) and gender (male and female). Respondents were screened
using the procedure described in Appendix~\ref{app:surveys}. The final
sample consisted of 1,108 unique respondents, each of whom rated an
average of 9.81 advertisements (SD = 0.61). Each advertisement was rated by an average of 12.8 unique respondents, with a standard deviation of 1.25 respondents. We targeted two raters from each partisanship–gender combination for every advertisement. In the end, each advertisement had an average of 2.13 raters per cell, and only eight advertisements lacked respondents in at least one partisanship–gender cell.

The dial position was recorded continuously throughout each advertisement, yielding a real-time rating for every video segment by taking the average of the ratings within each segment. Across all video segments, the mean rating was 50.8 and the standard deviation was 21.7, suggesting that respondents tended toward the midpoint of the 0--100 scale while exhibiting substantial variation in their evaluations. Indeed, respondents on average spent only about 30\% of each advertisement at the initial midpoint, and only 2.4\% of advertisement ratings (257 in total) exhibited no dial movement whatsoever. Moreover, none of the 1,108 respondents left the dial stationary across all of the advertisements they were assigned. These patterns indicate that respondents actively engaged with the task rather than leaving the dial fixed.

The substantive pattern of dial movements further suggests that respondents' evaluations reflected their political predispositions rather than random motion. Democratic respondents rated Democratic advertisements substantially more favorably than Republican advertisements (62.9 versus 26.8 on the 0--100 scale), while Republican respondents exhibited the mirror-image pattern (58.5 versus 41.6). Independents fell between the two partisan groups on both types of advertisement (53.1 and 37.1, respectively). 

The ordering across groups was perfectly monotonic: for Democratic
advertisements, mean ratings declined from Democrats (62.9) to
Independents (53.1) to Republicans (41.6), whereas for Republican
advertisements they declined from Republicans (58.5) to Independents
(37.1) to Democrats (26.8). The same pattern holds when attention is
restricted to advertisements featuring Joe Biden and Donald
Trump. Democrats rated Biden advertisements 35.5 points higher than
Trump advertisements, whereas Republicans rated Biden advertisements
19.6 points lower. Partisan differences of this magnitude provide
strong evidence that respondents were not merely moving the dial
randomly, but adjusting it in ways that reflected their political
predispositions.

\section{Methodological Framework} \label{sec:methods}

In this section, we introduce our proposed methodology for causal
inference with video data. Our framework extends the GPI methodology
of \cite{imai2026causal,imai2026leveraging} and \cite{nakamura2026genai}
to the {\it video-as-treatment} setting, where both the treatment and
outcome are continuously observed over time, motivated by the
empirical application in Section~\ref{sec::example}. However, the
proposed approach also applies more broadly to other dynamic treatment
settings, including text-based applications.

\subsection{Setup}

Consider a simple random sample of $N$ respondents drawn from a
population of interest. For each individual respondent
$i = 1, \ldots, N$, we randomly assign a video consisting of $S_i$
mutually exclusive and collectively exhaustive ordered segments, where
$S_i \in \cS := \{1,\ldots,s_{\max}\}$ represents the total number of
segments in the video assigned to individual $i$. 
The number of segments may vary across individuals, with $s_{\max} > 1$
denoting the maximum possible number of segments. 
Throughout the paper, we treat the segment count $S_i$ as fixed for each video, 
while allowing the video content within each segment to vary across videos.
Because videos of different lengths are randomly assigned to respondents, $S_i$ is a random variable. The random assignment of videos also ensures that each individual's exposure to video is unconfounded.

The length of segments may differ across applications, and the
appropriate granularity of segmentation depends on the substantive
context. In general, segments should be sufficiently fine-grained to
avoid aggregating multiple distinct actions into a single segment,
while remaining coarse enough to preserve meaningful variation in
information across segments.

Let $\bX_{is} \in \cX_s$ represent the $s$th segment of video assigned
to individual $i$, where $\cX_s$ denotes the support of the $s$th
segment $\bX_{is}$. For each individual, we measure the sequence of
outcomes $\overline{\bm Y}_i = (Y_{i1}, \cdots, Y_{iS_i})$ where
$Y_{is}$ represents the observed outcome for respondent $i$ for the
$s$th segment.  We adopt the potential outcome framework for
longitudinal causal inference \citep{neym:23,robins1986new, rubi:90}. Specifically, let $Y_{is}(\bx_1, \ldots, \bx_s)$ denote the
potential outcome at segment $s$ that would be realized if respondent
$i$ were exposed to the sequence of video segments
$(\bX_{i1}, \ldots, \bX_{is}) = (\bx_1, \ldots, \bx_s)$. We assume
that the observed outcome is given by
$Y_{is} = Y_{is}(\bX_{i1}, \ldots, \bX_{is})$. This consistency
assumption rules out spillover effects between respondents while
allowing for arbitrary carryover effects, so that the content of
earlier video segments may influence outcomes at later segments.

\begin{assumption}[Consistency]\label{consistency}
  For each interval $s=1,\ldots,S_i$, the observed outcome $Y_{is}$
  equals the potential outcome under the realized videos
  $\bX_{i1}, \ldots, \bX_{is}$, i.e.,
  $$
  Y_{is} =  Y_{is}(\bX_{i1}, \ldots, \bX_{is}).
  $$
\end{assumption}

\begin{assumption}[Random Assignment of Video]\label{randomization} The unstructured data $\bX_i$ is
  randomly assigned to each individual respondent so that
$$
   Y_{is}(\bx_1, \ldots, \bx_s) \ \indep \ \bX_{i} 
$$
holds for all $(\bx_1 , \ldots, \bx_s)$ given any $i=1,\ldots,N$ and $s=1,\ldots,S_i$.
\end{assumption}

We are interested in the causal effects of a sequence of features. For
simplicity, we focus on a binary treatment feature, denoted by
$W_{is}$, for segment $s$ of the video assigned to respondent $i$. We
assume that this treatment feature is an objective characteristic of
the underlying unstructured data and therefore takes the same value
for all respondents exposed to the same video segment. Following
\cite{imai2026causal}, we assume the existence of a mapping from each
video segment to the corresponding binary treatment indicator.

\begin{assumption}[Treatment feature]\label{treatment_feature}
  There exists a deterministic function $g_{W}: \cX_s \mapsto \{0,1\}$
  that maps each interval of the unstructured data $\bX_{is}$ to a
  binary treatment feature of interest, i.e.,
$$
W_{is} = g_{W}(\bX_{is}).
$$
\end{assumption}

Next, we define the confounding features, which comprise all features
of $\bX_{is}$ other than the treatment feature of interest $W_{is}$
that affect the outcome and future confounding features.  These confounding features are assumed to
lie in a much lower-dimensional space than the video segments
themselves.  Unlike previous GPI methods
\citep{imai2026causal,imai2026leveraging,nakamura2026genai}, however, we
allow the confounding features to vary across video segments for the
same individual. As noted earlier, our framework permits arbitrary
carryover effects, so that features of earlier segments may influence
outcomes measured at later segments. Specifically, for individual $i$,
segment $s$, and outcome $Y_{is'}$ with $s' \ge s$, we define the
corresponding confounding features as $\bU_{is}^{(s')}$. These
confounding features are assumed to be objective characteristics of
the assigned video. However, unlike the treatment feature $W_{is}$,
they are not directly observed and must instead be inferred from the
data.

\begin{assumption}[Confounding features]\label{confounding_feature}
There exists an unknown vector-valued deterministic function
  $\bg_{\bU_s}^{(s')}: \cX_s \mapsto \cU_s^{(s')}$ that maps each interval of
  unstructured data $\bX_{is} \in \cX_s$ to the confounding features
  $\bU_{is}^{(s')} \in \cU_s^{(s')}$, i.e.,
\begin{align*}
    \bU_{is}^{(s')} = \bg_{\bU_s}^{(s')}(\bX_{is})
\end{align*}
where $\mathrm{dim}(\bU_{is}^{(s')}) \ll \mathrm{dim}(\bX_{is})$. 
%$\bU_{is}$ captures all the relevant features of $\bX_{is}$ that influence outcome $Y_{is}$ such that
%\begin{align}
%    Y_{is} \ \indep \ \bX_{is} \ \mid \ W_{is}, \bU_{is}.
%\end{align}
\end{assumption}

Finally, we impose the following {\it sequential separability}
assumption \citep{nakamura2026genai}, which requires that the
treatment feature of each segment can be intervened upon without
altering the confounding features of that segment. In particular, the
confounding features must not be functions of the treatment feature,
thereby ruling out post-treatment bias.

\begin{assumption}[Sequential Separability]\label{separability} We have
\begin{align*}
    Y_{is}(\bx_1, \ldots, \bx_s)
    & \ = \ Y_{is}(g_W(\bx_1), \bg_{\bU_1}^{(s)}(\bx_1), \ldots, g_W(\bx_s), \bg_{\bU_s}^{(s)}(\bx_s) ) \\
    \bU_{is'}^{(s)}(\bx_1, \bx_2 \ldots, \bx_{s'-1}) & \ = \ \bU_{is'}^{(s)}(g_W(\bx_1), \bg_{\bU_1}^{(s)}(\bx_1), \ldots, g_W(\bx_{s'-1}), \bg_{\bU_{s'-1}}^{(s)}(\bx_{s'-1}))
\end{align*}
In addition, there exist no deterministic functions $\bg^\prime: \cX_s \to \cX_s^\prime$ and $\tilde{\bg}_{\bU_s}^{(s')}: \{0,1\} \times \cX_s^\prime \to \cU_s^{(s')}$, which satisfy $\bg_{\bU_s}^{(s')}(\bx) = \tilde{\bg}_{\bU_s}^{(s')} ( g_W(\bx), \bg^\prime(\bx))$ for all $\bx \in \cX_s$ and $\tilde{\bg}^{(s')}_{\bU_s} ( 1, \bg^\prime(\bx^\prime))\ne \tilde{\bg}^{(s')}_{\bU_s} (0, \bg^\prime(\bx^\prime))$ for some $\bx^\prime \in \cX_s$. 
\end{assumption}

Assumption~\ref{separability} consists of two components. The first
states that the potential outcome depends on each video segment
$\bX_{is}$ only through its treatment feature $W_{is}$ and confounding
features $\bU_{is}^{(s)}$. The second rules out the possibility that
the confounding features are determined by the treatment feature. We
express this condition in terms of deterministic functions rather than
stochastic relationships because, under the GPI framework, both the
treatment and confounding features are assumed to be deterministic
functions of the same underlying unstructured data. 

Unlike Assumptions~\ref{consistency}--\ref{confounding_feature}, the
second component of Assumption~\ref{separability} is not guaranteed by
the study design and may therefore be restrictive. Nevertheless, it
implies that the support of the treatment feature is independent of
the support of the confounding features at each segment and therefore
can be empirically tested \citep[see][for related results and
discussions in cross-sectional
settings]{wang_desiderata_2022,imai2026leveraging}.

\begin{proposition}[Sequential separability implies independent support]
\label{prop:seq_independent_support}
Suppose Assumption~\ref{separability} holds. Fix any
$1 \le s \le s' \le s_{\max}$. Define the joint support induced by the $s$th video segment as
\[
  \mathcal S_{WU,s}^{(s')}
  :=
  \left\{
  \left(g_W(\bx), \bg_{\bU_s}^{(s')}(\bx)\right):
  \bx \in \cX_s
  \right\}.
\]
Then, we have,
\[
  \mathcal S_{WU,s}^{(s')}
  =
  \mathcal W_s \times \mathcal U_s^{(s')}.
\]
\end{proposition}
See \cite{nakamura2026genai} for the proof.
Proposition~\ref{prop:seq_independent_support} establishes that, under
sequential separability, any feasible treatment value and any feasible
confounder value can jointly occur within a given segment. This is
precisely the support condition required for positivity conditional on
the confounding features, as it ensures that both treatment levels are
attainable for every confounder value.

\begin{figure}[t]
  \centering
    \includegraphics[width=1.0\linewidth]{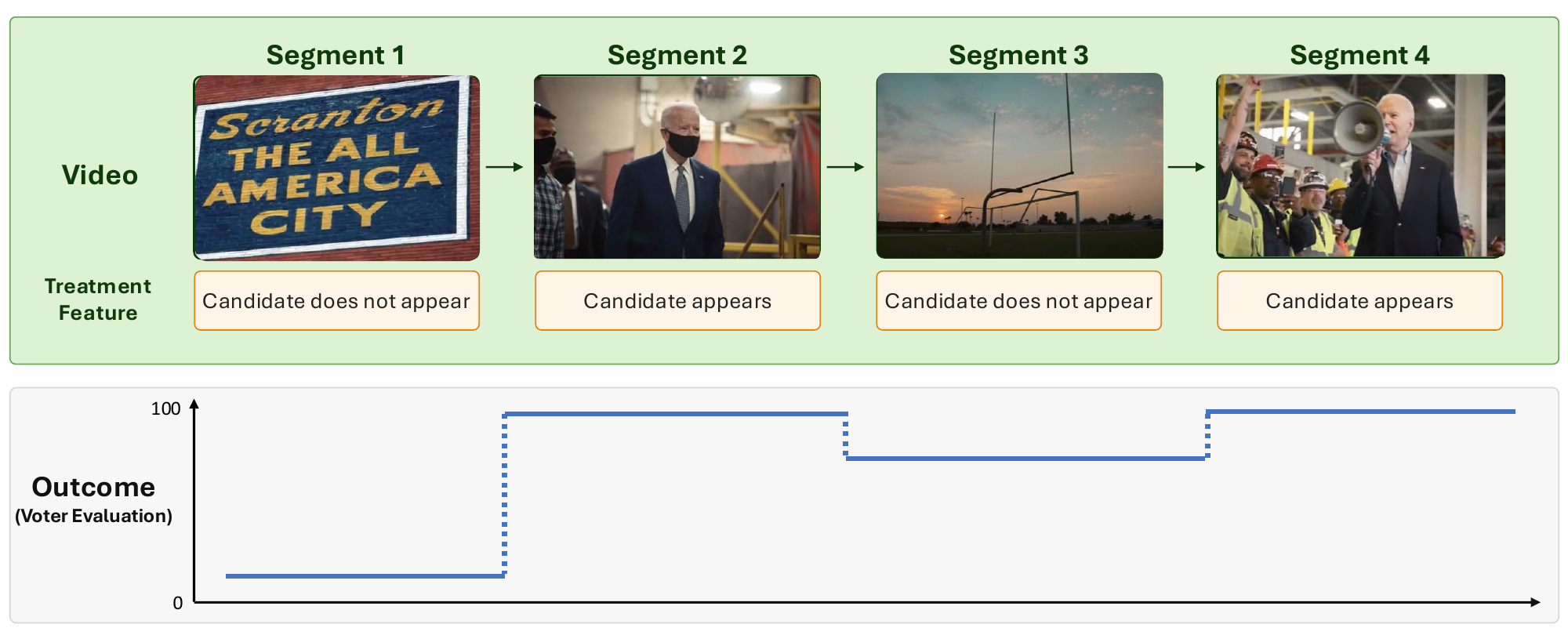}
    \caption{Schematic illustration of our experimental setup. We
      segment each video into short clips. For each segment, we
      observe the treatment feature (e.g., whether the candidate
      appears) and the outcome. Our goal is to estimate the causal
      effects of the sequence of treatment features on the sequence of
      outcomes, while adjusting for the unobserved confounding
      features of each segment. Importantly, our framework
      accommodates arbitrary carryover effects.}
    \label{fig:setup}
\end{figure}

Figure~\ref{fig:setup} schematically illustrates our experimental
setup. Importantly, the framework accommodates arbitrary carryover
effects, both short- and long-term. In our empirical application, for
example, a candidate's appearance in an early segment may affect a
respondent's evaluation at later segments, while earlier evaluations
may themselves shape subsequent evaluations. As in the standard
framework for longitudinal causal inference with time-varying
treatments (e.g., \citealt{robi:hern:brum:00, blackwell_how_2018}),
our approach allows for such dynamic effects, thereby broadening its
applicability.  The key distinction from the standard setting,
however, is that the confounding features are not directly observed
and therefore must be learned from the video data.

\subsection{Leveraging Deep Generative Models}

The primary challenge for causal inference with video data is its
extremely high dimensionality. Even a short video clip contains
numerous frames, each consisting of thousands of pixels, yielding a
feature space that far exceeds any feasible sample size for
statistical analysis. As a result, any viable analytical strategy must
reduce dimensionality while preserving the information relevant for
causal inference.

The GPI framework addresses this challenge by leveraging a deep
generative model to represent unstructured data, rather than relying
on an external embedding whose information loss is generally unknown
\citep{imai2026causal,imai2026leveraging, nakamura2026genai}. Specifically,
GPI generates---or, in the case of existing videos, reproduces---each
clip using an open-source generative model and extracts the model's
internal representations. Because these representations are used to
generate the video itself, they contain the information necessary to
reconstruct it, thereby providing a principled low-dimensional
representation for subsequent statistical inference. Indeed,
researchers can eliminate information loss altogether by using the
reproduced video clips in experiments. In our empirical application,
we demonstrate that videos regenerated using a state-of-the-art
video generation model are visually indistinguishable from the original
videos.

The key object is the internal representation from which the video is
generated. We make this notion precise with the following general
definition of a deep generative model. 

\begin{definition}[Deep Generative Model]\label{deep_use}
  \spacingset{1} A deep generative model is the following
  probabilistic model that takes an input $\bP_{i}$ (e.g., prompts or
  video clips) and generates the treatment object $\bX_{i}$ as an
  output:
$$
\begin{aligned}
&\P(\bX_{i}  \mid \bh_{\bgamma}(\bR_{i}))\\
&\P(\bR_{i} \mid \bP_{i})
\end{aligned}
$$
where $\bR_{i} \in \cR \subset \R^{D_R}$ denotes an observable
internal representation of $\bX_{i}$ contained in the model and
$\bh_{\bgamma}(\bR_{i})$ is a deterministic function parameterized by
$\bgamma$ that completely characterizes the conditional distribution
of $\bX_{i}$ given $\bR_{i}$.
\end{definition}
Under this definition, $\bR_i$ is a lower-dimensional representation
of $\bX_i$ that corresponds to a latent representation within the deep
generative model. This formulation encompasses a broad class of video
generation models, including diffusion-based architectures
\citep[e.g.,][]{agarwal2025cosmos}. As discussed earlier, when the
goal is to analyze existing video data rather than generate new
content, we first reproduce the videos using an appropriately prompted
deep generative model and then extract the corresponding latent
representations.

We assume that each video segment $\bX_{is}$ is generated directly
from a corresponding latent representation, denoted by $\bR_{is}$. In
addition, we assume that the output layer of the deep generative model
is a deterministic function of $\bR_{is}$. Under these assumptions,
the relevant low-dimensional features of the video segment, i.e., the
treatment feature $W_{is}$ and the confounding features
$\{\bU_{is}^{(s')}\}_{s' \ge s}$, can be expressed as deterministic
functions of the latent representation $\bR_{is}$. We formalize these
assumptions below.

\begin{assumption}[Factorized Deterministic Decoding]\label{det_dec} The output
  layer of a deep generative model is deterministic for a given video
  interval. That is, given the number of segments $S_i$,
\begin{align*}
   \P(\bX_{i} \mid \bh_{\bgamma}(\bR_i)) = \prod_{s = 1}^{S_i} \P(\bX_{is} \mid \bh_{\bgamma}(\bR_{is})) 
\end{align*}
and $\P(\bX_{is} \mid \bR_{is})$ is degenerate for any respondent $i$ and segment $s$.
\end{assumption}
State-of-the-art video generation models are typically
diffusion-based architectures that generate an entire video clip
jointly rather than frame by frame. To ensure that
Assumption~\ref{det_dec} holds at the segment level, we partition each
video clip into smaller segments, reproduce each segment $\bX_{is}$
independently using a deep generative model, and extract the latent
representation from the final layer immediately preceding the
decoder. 

\begin{figure}[t!]
    \centering
    \includegraphics[width=1.0\linewidth]{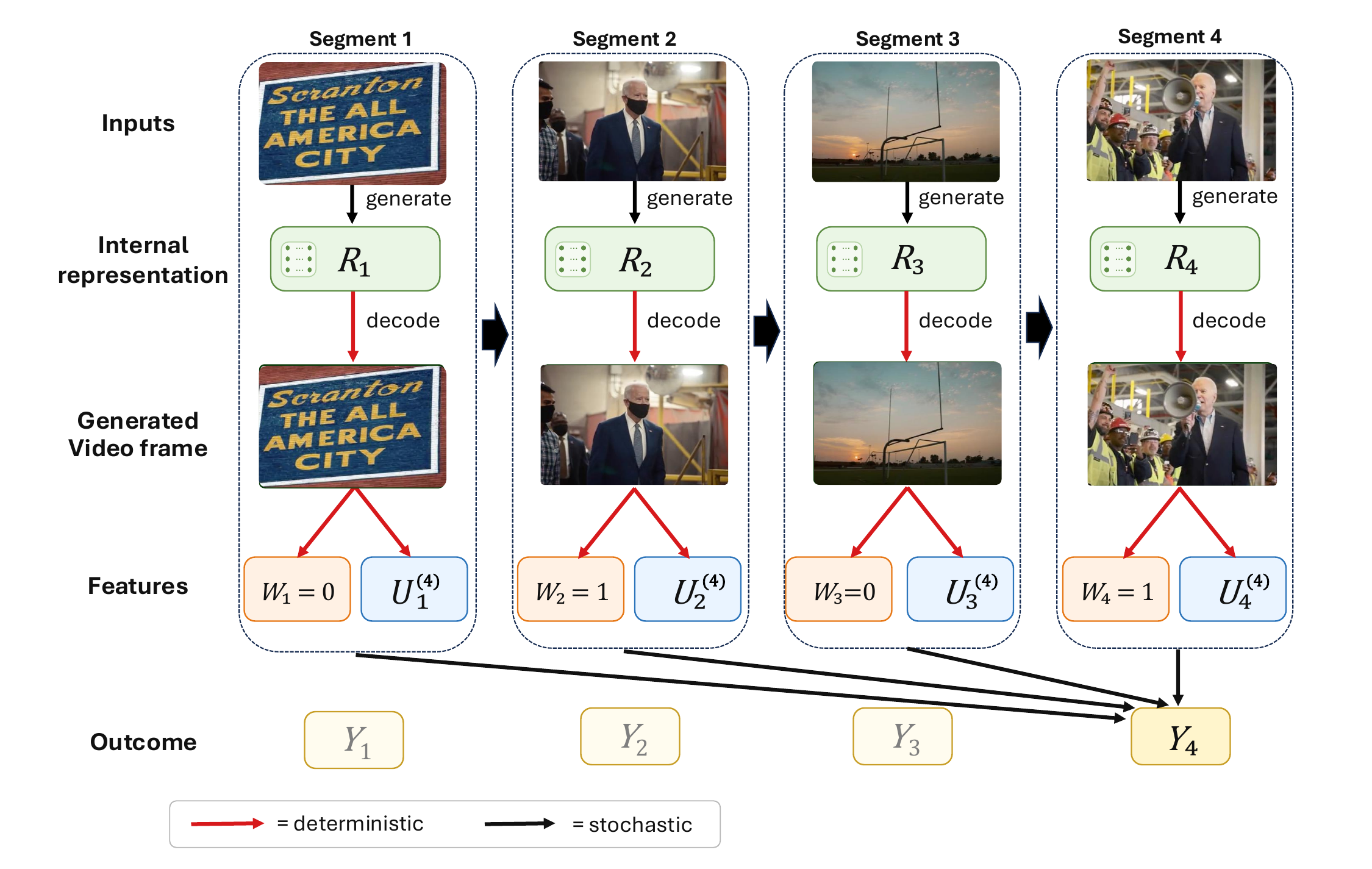}
    \caption{Illustration of the assumed data generating process when
      we are interested in the outcome $Y_4$.  A red arrow with double
      lines represents a deterministic causal relation while a black
      arrow with a single line indicates a possibly stochastic
      relationship. The same framework applies to any other outcome
      $Y_s$, though the relevant confounding features may differ
      across the target outcomes.}
\label{dag}
\end{figure}

Figure~\ref{dag} illustrates our data-generating process and the assumptions
described above. In this illustration, an arrow with red double lines represents
a deterministic causal relation, whereas an arrow with a single black line
represents a possibly stochastic causal relation. At each segment
$s \in \{1, \cdots, s_{\max}\}$, we regenerate the entire video using the deep
generative model and extract the internal representation $\bR_{is}$. Each video
segment can influence the outcome through both the observed treatment features
and the unknown confounding features. However, under
Assumptions~\ref{confounding_feature}~and~\ref{det_dec}, the
confounding features are a deterministic function of the internal representation,
which is lower-dimensional and computationally feasible to model with finite
data.
Since the confounding features defined in
Assumption~\ref{confounding_feature} capture all the relevant features
of each video segment that influence the corresponding outcome and
future confounding features together with the treatment feature, they
are assumed to satisfy
\begin{equation}
\begin{aligned}
      Y_{is} \ & \indep \ \{\bR_{i1}, \cdots, \bR_{is}\} \mid W_{i1},
  \cdots, W_{is}, \bU_{i1}^{(s)}, \cdots, \bU_{is}^{(s)}, S_i \geq s \\
  \bU_{i,s'+1}^{(s)} & \indep \
  \{\bR_{i1}, \cdots, \bR_{is'}\} \mid W_{i1}, \cdots, W_{is'},
  \bU_{i1}^{(s)}, \cdots, \bU_{is'}^{(s)}, S_i \geq s \label{eq:observe_independence}
\end{aligned}
\end{equation}
for any $s \in \cS$ and $s' < s$.

\subsection{Causal Estimand under Dynamic Stochastic Interventions} \label{sec::estimand}

We conduct causal inference under a stochastic intervention
\citep[e.g.,][]{muno:vand:12, kennedy_nonparametric_2019,
  papa:etal:22}. We focus on stochastic rather than
deterministic interventions because the latter are often prone to
violations of the positivity assumption. In our application, for
example, the treatment feature indicates whether a candidate appears
on screen, a characteristic that typically exhibits strong temporal
persistence across adjacent segments. If a candidate appears in one
segment, they are likely to appear in the next.  Thus, a deterministic
intervention that forces the candidate to appear in every segment
except one would induce a highly implausible treatment sequence that
may have little or no support in the observed data, thereby violating
the positivity assumption required for identification.

Formally, we consider the following dynamic stochastic intervention
developed in \cite{nakamura2026genai} based on
\cite{kennedy_nonparametric_2019}:
\begin{align}
q_s(\delta_s; \overline{\bm{w}}_{s-1}) := 
\frac{\delta_s p_s(\overline{\bm{w}}_{s-1}) }{\delta_s
  p_s(\overline{\bm{w}}_{s-1}) + 1 -
  p_s(\overline{\bm{w}}_{s-1})}, \label{def:stochastic}
\end{align}
for each segment $s=1,\ldots,S_i$, where
$p_s(\overline{\bm{w}}_{s-1}) = \P(W_{is} = 1 \mid
\overline{\bm{W}}_{i,s-1} = \overline{\bm{w}}_{s-1}, S_i \geq s)$ denotes the
observed treatment assignment probability. Note that
$\overline{\bm{W}}_{i,s-1} = (W_{i1}, \ldots, W_{i,s-1})$ represents
the past treatment history, whose value is given by the vector
$\overline{\bm{w}}_{s-1} = (\bm{w}_1, \ldots, \bm{w}_{s-1})$ for
$s > 1$ with $\overline{\bm{W}}_{i0} = \emptyset$.

Here, the user-specified incremental parameter
$\delta_s \in (0, \infty)$ determines the extent to which the
counterfactual treatment assignment probabilities diverge from the
observed ones, where a greater value of $\delta_s$ corresponds to a
policy that assigns treatment with a higher probability. Indeed,
$\delta_s$ equals the odds ratio:
$$\delta_s \ = \ \frac{q_s(\delta_s; \overline{\bm{w}}_{s-1}) }{1 -
  q_s(\delta_s; \overline{\bm{w}}_{s-1}) } \cdot \frac {1 -
  p_s(\overline{\bm{w}}_{s-1}) }{p_s(\overline{\bm{w}}_{s-1}) }.$$
This dynamic stochastic intervention is identical to the incremental
propensity score intervention of \citet{kennedy_nonparametric_2019},
except that we use the observed treatment probability conditional on
treatment history rather than the propensity score. This modification
is necessary since we do not observe the confounding features
directly.

Under this stochastic intervention, we consider the average potential
outcome aggregated over all segments, which is formally defined as
\begin{align}
    \bar\Psi(\bm{\delta}) := \E\Biggl[ \sum_{s = 1}^{S_i} \Biggl\{ \int_{\mathcal{W}^{s}} Y_{is}(W_{i1} = w_1, \ldots, W_{is} = w_s) 
  \prod_{s'= 1}^{s} d Q_{s'}(w_{s'}; \overline{\bm{w}}_{s'-1}, \delta_{s'}) \Biggr\} \Biggr], \label{target_quantity}
\end{align}
where $\bm\delta = (\delta_1, \cdots, \delta_{s_{\max}})$ and
$$
dQ_s(w_s; \overline{\bm{w}}_{s-1}, \delta_s) = \frac{w_s \delta_s
  p_s(\overline{\bm{w}}_{s-1}) + (1 - w_s) (1 -
  p_s(\overline{\bm{w}}_{s-1}) ) }{\delta_s p_s(\overline{\bm{w}}_{s-1}
  )+ 1 - p_s(\overline{\bm{w}}_{s-1})}.
$$
We omit confounding features in Equation~\eqref{target_quantity} from the potential outcome for notational simplicity. 
%Throughout the paper, we use the following notation
%\begin{align*}
%    \bm\delta_{\underline{s}} \ := \ \bm\delta_{(s+1):s_{\max}}, \qquad
%    \bm\delta_{\overline{s}} \ := \ \bm\delta_{1:s}, \qquad
%    \bm\delta \ = \ (\bm\delta_{\overline{s}}, \bm\delta_{\underline{s}}), \qquad
%    \bm\delta_{-s} \ := \ (\delta_1, \ldots, \delta_{s-1}, \delta_{s+1}, \ldots, \delta_{s_{\max}}).
%\end{align*}

This estimand summarizes the expected cumulative outcome under a
dynamic stochastic intervention indexed by the incremental parameter
$\bm\delta$. Under this intervention, the treatment features are
generated from the modified treatment assignment mechanism, whereas
the confounding features are left to evolve according to their natural
data-generating process. The inner integral averages the potential
outcome over possible treatment histories, with each history weighted
by the intervention distribution that increasingly favors treatment as
$\delta_s$ becomes larger. The outer expectation then averages this
quantity over the distribution of confounding histories, potential
outcomes, and the number of observed segments $S_i$.

The stochastic intervention defined in Equation~\eqref{def:stochastic}
mitigates positivity violations by avoiding treatment values that are
unsupported by the observed data. Specifically, whenever the observed
treatment probability equals 0 or 1, the intervention probability is
set to the same value regardless of $\delta$. In our application, for
example, if a candidate always appears (or never appears) at a given
segment in the observed data, the intervention preserves that
probability rather than imposing a treatment path that lacks empirical
support.

As a result, identification requires only a relative form of
positivity. Formally, we denote the history of confounding features
for outcome $Y_{is}$ up to segment $s$ by
$\overline{\bm{U}}_{is}^{(s)} = (\bU_{i1}^{(s)}, \ldots,
\bU_{is}^{(s)})$ for $s \ge 1$ and
$\overline{\bm{U}}^{(s)}_{i0}=\emptyset$.  Then, we can define the
propensity score at segment $s$ as
$\pi_s(\overline{\bm{W}}_{i,s-1}, \overline{\bm{U}}^{(s)}_{is}) :=
\P(W_{is} = 1 \mid \overline{\bm{W}}_{i,s-1},
\overline{\bm{U}}^{(s)}_{is})$, which represents the conditional
probability of treatment assignment given the treatment up to $s-1$ and confounder
histories up to interval $s$. We formalize this assumption as
follows.

\begin{assumption}[Bounded relative overlap]\label{overlap} There exists a constant $c > 0$ such that 
\begin{align*}
\pi_s(\overline{\bm{w}}_{s-1}, \overline{\bm{u}}_{s}) \ge c \cdot p_s(\overline{\bm{w}}_{s-1}) \quad \text{and} \quad  1 - \pi_s(\overline{\bm{w}}_{s-1}, \overline{\bm{u}}_{s}) \ge c \cdot (1 -  p_s(\overline{\bm{w}}_{s-1}))
\end{align*}
for all $s=1,\ldots,S_i$, $i=1,\ldots,N$,
$\overline{\bm{w}}_{s-1} = (\bm{w}_1, \ldots, \bm{w}_{s-1}) \in
\mathcal{W}^{s-1}$, and
$\overline{\bm{u}}_{s} = (\bm{u}_1, \ldots, \bm{u}_{s}) \in
\bar{\mathcal{U}}^{(s)}$ where
$\bar{\mathcal{U}}^{(s)} = \prod_{s^\prime = 1}^{s}
\cU^{(s)}_{s^\prime}$. 
\end{assumption}
Importantly, Assumption~\ref{overlap} is closely connected to
Proposition~\ref{prop:seq_independent_support}, as both are statements
about positivity given the confounding features at each segment. This
connection suggests a practical diagnostic for our non-design-based
requirements: although Assumptions~\ref{separability}
and~\ref{overlap} are not guaranteed by design, we can still assess
them by inspecting the distribution of the estimated propensity
scores.

\subsection{Identification}

We establish identification of the average potential outcome under stochastic intervention in Equation~\eqref{target_quantity}. 
This result is a natural extension
of Theorem~1 of \cite{nakamura2026genai}, which identifies the causal
effect on a single outcome at a given segment under the same dynamic
intervention. For a single outcome, identification holds because
adjusting for the full history of latent confounding features suffices
to block all backdoor paths between the treatment sequence and the
outcome of interest (see Figure~\ref{dag}). We can thus establish identification for each segment-specific outcome separately and then aggregate these identified quantities across segments to obtain Equation~\eqref{target_quantity}. The key challenge here is that the confounding features are latent and must be inferred from the data. 

Suppose first that the confounding features
$\overline{\bU}_{is}^{(s)}$ were observed. Under
Assumptions~\ref{consistency}--\ref{overlap}, the sequential g-formula
identifies the average potential outcome at segment $s$ among
respondents with $S_i \geq s$ as
\begin{equation}
\begin{aligned}
\Psi_s(\bm\delta)
=\;&
\int_{\bar{\mathcal U}^{(s)}}
\int_{\mathcal W^s}
\E\left[
Y_{is}
\mid
\overline{\bW}_{is}=\overline{\bm w}_s,
\overline{\bU}_{is}^{(s)}=\overline{\bm u}_s,
S_i\geq s
\right]
\\
&\quad\times
\prod_{t=1}^{s}
\left\{
dQ_t(w_t;\overline{\bm w}_{t-1},\delta_t)
\,dF\left(
\bm u_t
\mid
\overline{\bW}_{i,t-1}=\overline{\bm w}_{t-1},
\overline{\bU}_{i,t-1}^{(s)}=\overline{\bm u}_{t-1},
S_i\geq s
\right)
\right\}.
\end{aligned}
\label{eq:gformula_U}
\end{equation}
Applying this
identification argument separately to each outcome and aggregating
across segments yields
\begin{align*}
\bar\Psi(\bm\delta)
=
\sum_{s=1}^{s_{\max}}
\P(S_i\geq s)\Psi_s(\bm\delta).
\end{align*}

In our setting, however, the confounding features are not observed.
Instead, under Assumptions~\ref{confounding_feature}
and~\ref{det_dec}, they are unknown deterministic functions of the
observed internal representations $\bR_{it}$. Identification does not
require recovering these functions exactly. Rather, for each target
outcome $Y_{is}$, we seek low-dimensional functions called \emph{deconfounder} 
$\boldf_t^{(s)}:\mathcal R\mapsto\mathcal F\subset\R^{D_f}$,
$t=1,\ldots,s$, that can replace the confounding features in the
sequential g-formula without changing its value. These
representations need only preserve the conditional means relevant
for identification, rather than all information contained in the
confounding features.

To formalize this idea, we first define the outcome model given the treatment and deconfounder history as
\begin{align}
\mu_{s}
\left(
\overline{\bm w}_{s},
\boldf^{(s)}_1(\bR_{i1}),\ldots,\boldf^{(s)}_{s}(\bR_{is})
\right)
&:=
\E[Y_{is} \mid \overline{\bW}_{is} = \overline{\bm w}_{s}, \boldf_1^{(s)}(\bR_{i1}), \ldots,
\boldf_{s}^{(s)}(\bR_{is}), S_i \geq s ]. \label{eq:outcome_model}
\end{align}
We also define the conditional mean function under the stochastic intervention,
\begin{equation}
\begin{aligned}
&m_{s' | s}
\left(
\overline{\bm w}_{s'},
\boldf_1(\bR_{i1}),\ldots,\boldf_{s'}(\bR_{is'});
\bm\delta_{(s'+1):s}
\right) \\
:= \ & \int_{\mathcal R^{S_i-s'}} \int_{\mathcal W^{S_i-s'}} \ \mu_{s}
\left(
\overline{\bm w}_{s},
\boldf^{(s)}_1(\bR_{i1}),\ldots,\boldf^{(s)}_{s}(\bR_{is})
\right)\\
&\qquad \times  \prod_{t = s'+1}^{s} d Q_{t}
(w_{t};\overline{\bm w}_{t-1},\delta_{t}) dF(\bR_{it} \mid \overline{\bm{W}}_{i,t-1} = \overline{\bm{w}}_{t-1},  \boldf_1^{(s)}(\bR_{i1}), \cdots, \boldf^{(s)}_{t-1}(\bR_{i,t-1}), S_i\geq s)
\end{aligned} \label{eq:recursive_m}
\end{equation}
for $s'=0,\ldots,s-1$. Substantively, $m_{s'|s}$ represents the expected outcome at $s$th segment when all future treatments from the ($s'+1$)th segment onward are generated according to the stochastic intervention, whereas $\mu_s$ represents the observed outcome model under deconfounder and treatment history. Thus, if the functions
$\boldf_{s'}^{(s)}(\bR_{is'})$ retain the same confounding information as
$\bU_{is'}^{(s)}$, the resulting $m_{0\mid s}$ identifies the average potential outcome for segment $s$.

Following \cite{nakamura2026genai}, we characterize a deconfounder
through mean-independence conditions for the outcome regression
$\mu_{s}$ and the intermediate outcome regressions $m_{s'\mid s}$
appearing in the sequential g-formula.

\begin{definition}[Deconfounder]
A collection of functions
$\boldf_s^{(s')}:\mathcal{R}\mapsto\mathcal{F}\subset\R^{D_f}$ is a deconfounder if it satisfies the following
mean-independence conditions:
\begin{equation}
\begin{aligned}
\E[Y_{is} \mid \boldf_1^{(s)}(\bR_{i1}), \ldots,
\boldf_{s}^{(s)}(\bR_{is}), \overline{\bW}_{is}, S_i \geq s] = \E[Y_{is} \mid
\boldf_{1}^{(s)}(\bR_{i1}),\ldots,\boldf_{s}^{(s)}(\bR_{is}),
\overline{\bW}_{is}, \overline{\bR}_{is}, S_i \geq s] \label{deconfounder_independence}
\end{aligned}
\end{equation}
and, for any pair of deconfounder functions $\boldf_{s'}^{(s)}$ and
$\tilde{\boldf}_{s'}^{(s)}$ (which may coincide so that
$\boldf_{s'}^{(s)}=\tilde{\boldf}_{s'}^{(s)}$),
\begin{equation}
\begin{aligned}
&\E\left[m_{s'+1 \mid s}\left(\overline{\bm{w}}_{s'+1},
  \tilde{\boldf}_1^{(s)}(\bR_{i1}), \ldots,
  \tilde{\boldf}_{s'+1}^{(s)}(\bR_{i,s'+1});
\bm\delta_{(s'+1):s}\right) \ \middle| \
  \boldf_1^{(s)}(\bR_{i1}), \ldots, \boldf_{s'}^{(s)}(\bR_{is'}),
  \overline{\bW}_{is'}, S_i \geq s\right]\\
&= \E\left[m_{s'+1 \mid s}\left(\overline{\bm{w}}_{s'+1},
  \tilde{\boldf}_1^{(s)}(\bR_{i1}), \ldots,
  \tilde{\boldf}_{s'+1}^{(s)}(\bR_{i,s'+1});
\bm\delta_{(s'+1):s}\right) \ \middle| \
  \boldf_1^{(s)}(\bR_{i1}), \ldots, \boldf_{s'}^{(s)}(\bR_{is'}),
  \overline{\bR}_{is'}, \overline{\bW}_{is'}, S_i \geq s\right],
  \label{pseudo_outcome_independence}
\end{aligned}
\end{equation}
for any $s\in\cS$ with $s>1$, $s'=1,\ldots,s-1$,
$\overline{\bm w}_{s'}\in\mathcal W^{s'}$, and any value of the
incremental parameter $\bm\delta\in\R_+^{s_{\max}}$, where
$\overline{\bR}_{is'}=(\bR_{i1},\ldots,\bR_{is'})$ and
$\overline{\bR}_{i0}=\emptyset$. 
%Here, $m_{s+1}$ denotes the outcome models, which are recursively defined as
%\begin{align*}
%m_s
%\left(
%\overline{\bm w}_s,
%\boldf_1(\bR_{i1}),\ldots,\boldf_s(\bR_{is}); \bm\delta_{\underline{s}}
%\right)
%&:=
%\E\left[
%Y_i
%\mid
%\overline{\bW}_{is}=\overline{\bm w}_s,
%\boldf_1(\bR_{i1}),\ldots,\boldf_s(\bR_{is}),
%S_i=s
%\right],
%\end{align*}
%and, for $s'=1,\ldots,s-1$, 
%\begin{equation*}
%\begin{aligned}
%&m_{s'}
%\left(
%\overline{\bm w}_{s'},
%\boldf_1(\bR_{i1}),\ldots,\boldf_{s'}(\bR_{is'});
%\bm\delta_{\underline{s'+1}}
%\right)\\
%&:=
%\E\left[
%\int_{\mathcal{W}}
%m_{s'+1}
%\left(
%\overline{\bm w}_{s'+1},
%\tilde \boldf_1(\bR_{i1}),\ldots,
%\tilde \boldf_{s'+1}(\bR_{i,s'+1});
%\bm\delta_{(s'+2):s}
%\right)
%dQ_{s'+1}
%(w_{s'+1};\overline{\bm W}_{is'},\delta_{s'+1})
%\right.\\
%&\hspace{3in}
%\Bigr | \ 
%\overline{\bW}_{is'}=\overline{\bm w}_{s'},
%\boldf_1(\bR_{i1}),\ldots,
%\boldf_{s'}(\bR_{i,s'}),
%S_i=s \biggr ].
%\end{aligned}
%\end{equation*}    
\end{definition}
Intuitively, a deconfounder retains the information needed for the
conditional means $\mu$ and $m_{s'|s}$ in the sequential g-formula. The first condition,
Equation~\eqref{deconfounder_independence}, requires that, conditional
on the treatment and deconfounder histories, adding the full internal
representation history does not change the conditional mean of
$Y_{is}$. The second condition,
Equation~\eqref{pseudo_outcome_independence}, imposes the analogous
requirement at each earlier stage of the backward recursion.

To connect these conditions to the latent confounding features,
recall that both $\bU_{it}^{(s)}$ and
$\boldf_t^{(s)}(\bR_{it})$ are deterministic functions of
$\bR_{it}$. The conditional independence of $Y_{is}$ from
$\overline{\bR}_{is}$ given the treatment and confounding histories,
together with Equation~\eqref{deconfounder_independence}, implies that
\begin{align*}
\E\left[
Y_{is}
\mid
\overline{\bW}_{is},
\overline{\bU}_{is}^{(s)},
S_i\geq s
\right] &=
\E\left[
Y_{is}
\mid
\overline{\bW}_{is},
\boldf^*_1(\bR_{i1}),
\cdots,
\boldf^*_s(\bR_{is}),
S_i\geq s
\right]\\
&= \E\left[
Y_{is}
\mid
\overline{\bW}_{is},
\overline{\bR}_{is},\boldf^*_1(\bR_{i1}),
\cdots,
\boldf^*_s(\bR_{is}),
S_i\geq s
\right]\\
&= \E\left[
Y_{is}
\mid
\overline{\bW}_{is},
\overline{\bR}_{is},\boldf_1^{(s)}(\bR_{i1}),\ldots,
\boldf_s^{(s)}(\bR_{is}),
S_i\geq s
\right]\\
& =
\E\left[
Y_{is}
\mid
\overline{\bW}_{is},
\boldf_1^{(s)}(\bR_{i1}),\ldots,
\boldf_s^{(s)}(\bR_{is}),
S_i\geq s
\right],
\end{align*}
where the first and third equalities follow from the fact that the confounding features are deterministic functions of the internal representations, the second equality follows from Equation~\eqref{eq:observe_independence}, and the final equality follows from Equation~\eqref{deconfounder_independence}.
Therefore, conditioning on the deconfounder history yields the same
outcome regression as conditioning on the latent confounding
history. The second mean-independence condition similarly extends this
information-preservation requirement to the intermediate
conditional means in each step of backward recursion.
Importantly, although many functions may satisfy the same mean-independence conditions, any such function yields the same identification formula.

The following theorem establishes that these mean-independence
conditions, together with bounded relative overlap and the
maintained causal assumptions, allow us to identify the quantity of interest.

\begin{theorem}[Identification for Temporally Aggregated Average Potential Outcome]\label{theorem:identification}
  Under Assumptions \ref{consistency}--\ref{overlap}, for each
  $s \in \cS$, there exist deconfounder functions
  $\boldf_{s'}^{(s)}: \cR \mapsto \cF \subset \R^{D_f}$,
  $s' = 1, \ldots, s$ that satisfy bounded relative overlap.
By adjusting for such a deconfounder, we can nonparametrically
identify the average outcome under the stochastic intervention defined
in Equation~\eqref{target_quantity},
\begin{align*}
\bar\Psi(\bm\delta)
& \ = \  \E\left[ \sum_{s = 1}^{S_i} \int_{\cR^{s}} \int_{\mathcal{W}^{s}}
                \E\left[Y_{is} \mid
                \overline{\bW}_{is} = \overline{\bm{w}}_{s}, \
                \{\boldf^{(s)[s]}_{s''}(\bR_{is''})\}_{s''=1}^{s}, S_i \geq s\right]
                \right. \\
  & \left. \hspace{0.2in} \times \prod_{s' = 1}^{s} \left\{ d Q_{s'}(w_{s'}; \overline{\bm{w}}_{s'-1}, \delta_{s'}) dF(\bR_{is'} \mid \overline{\bm{W}}_{i,s'-1} = \overline{\bm{w}}_{s'-1}, \{\boldf^{(s)[s'-1]}_{s''}(\bR_{is''})\}_{s''=1}^{s'-1}, S_i \geq s) \right\} \right],
\end{align*}
where the superscript $[t]$ in $\boldf_{s'}^{(s)[t]}$ is introduced to
indicate that the deconfounder may differ across the levels of the
recursively defined outcome models, i.e., we can have
$\boldf_{s'}^{(s)[t]}(\bR_{is'}) \neq \boldf_{s'}^{(s)[r]}(\bR_{is'})$
with $t \neq r$.
\end{theorem}
The proof is omitted as it is similar to that of Theorem~1 of
\cite{nakamura2026genai}.

\subsection{Estimation and Inference}

Given this identification result, we now turn to estimation and
inference.  The central challenge is estimating the deconfounder
characterized by Assumption~\ref{separability}.  Namely, the
deconfounder cannot be a function of the treatment feature within the
same segment, together with the mean-independence conditions in
Equations~\eqref{deconfounder_independence}
and~\eqref{pseudo_outcome_independence}. 
%In cross-sectional settings, the deconfounder can be learned with a neural network architecture that encodes this mean-independence relationship \citep{imai2026causal}. With video, however, the difficulty is that we must learn the deconfounder $\boldf_{s'}^{(s)[t]}$ at the $t$-th outcome regression for every target outcome $s \in \cS$ and every internal representation $\bR_{is'}$, while satisfying Equation~\eqref{pseudo_outcome_independence}.

\begin{figure}[t]
\centering
\includegraphics[width=1.0\linewidth]{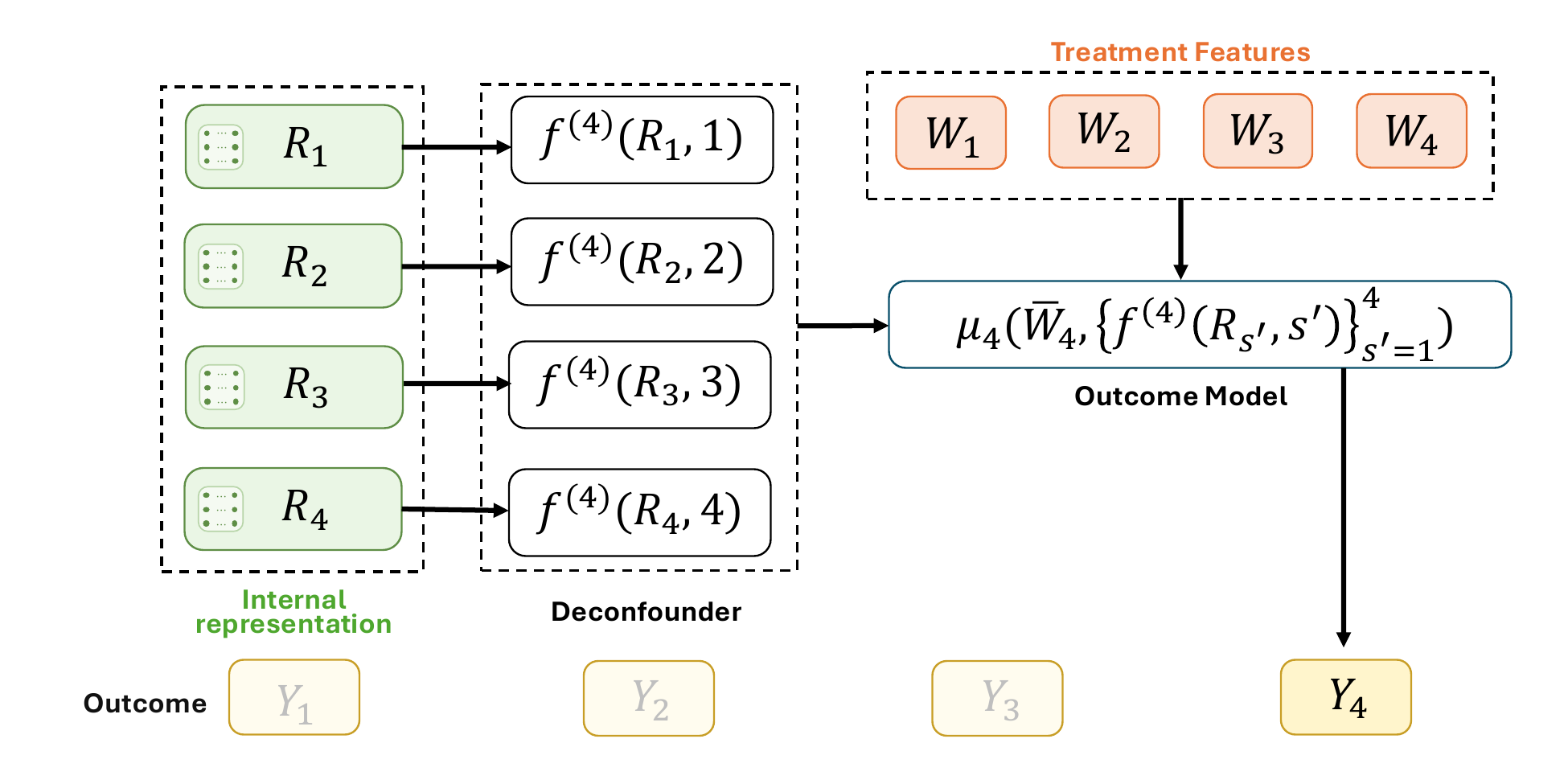}
\caption{Diagram Illustrating the Proposed Neural Network Architecture
  for the case with four segments. The proposed model takes an
  internal representation at each segment $\bR_{is'}$ as an input and
  finds the deconfounder $\boldf^{(s)}(\bR_{is'})$, which is a
  lower-dimensional representation of $\bR_{is'}$. The set of
  deconfounders at all segments and the treatment history
  $\overline{\bm{W}}_{is}$ are then used to predict the outcome
  $Y_{is}$.
}
\label{architecture}
\end{figure}

We estimate a single neural network, shared across earlier segments for a given target outcome, to estimate the deconfounder and outcome model at each step of the backward recursion. As an illustration, consider the estimation of $\Psi_s$, the average potential outcome at segment $s$. We begin the backward recursion with $\mu_{s}$ and estimate this outcome model using the neural network architecture illustrated in Figure~\ref{architecture}. In this architecture, the history of internal representations predicts the outcome only through the treatment and deconfounder histories, thereby encoding the mean-independence condition in Equation~\eqref{deconfounder_independence}. In addition, the deconfounder is not a function of the treatment history, which ensures that it does not contain post-treatment information.
Note that to prevent the parameter explosion, we reparametrize deconfounder as $\boldf^{(s)}: \cS \times \cR \mapsto \cQ$, which takes the internal
representation $\bR_{is}$ and the segment indicator as inputs. The use
of segment indicator allows the deconfounder to vary for each input.  
We fit this
model by minimizing the following squared loss,
\begin{align}
    \{\hat\blambda_s, \hat\btheta_s\} = \argmin_{\btheta_s, \blambda_s} \; \frac{1}{N_s} \ \sum_{i = 1}^N \ \mathbbm{1}\{S_i \geq s\} \left\{ Y_{is} - \mu_{s}(\overline{\bm{W}}_{i s}, \{\boldf^{(s)}(\bR_{is'},s';
  \blambda_s)\}_{s'=1}^{s}; \btheta_s) \right\}^2\label{eq:minimization}
\end{align}
where $N_s = \sum_{i = 1}^N \mathbbm{1}\{S_i \geq s\}$.
Here, we make the network parameters
explicit with $\bm\lambda_s$ collecting the parameters of the
deconfounder $\boldf^{(s)}$, and $\btheta_s$ denoting the parameters
of the nuisance function.

Once we estimate the terminal outcome model, we recursively estimate
$m_{s' \mid s}$ for $s' = s-1, \cdots, 0$ so that the deconfounder
also satisfies Equation~\eqref{pseudo_outcome_independence}.
Specifically, we first rewrite Equation~\eqref{eq:recursive_m} as
\begin{align*}
    m_{s'\mid s}(\overline{\bm{W}}_{is'}, \{\boldf^{[s]}(\bR_{is'},s'')\}_{s''=1}^{s'}, \bm\delta_{(s'+1):s})
  \ = \
    \E[\tilde{Y}_{is'} \mid \overline{\bm{W}}_{is'}, \{\boldf^{[s]}(\bR_{is''}, s'')\}_{s''=1}^{s'}, S_i \geq s].
\end{align*}
where $\tilde Y_{is'}$ is the pseudo-outcome defined as
\begin{equation}
    \begin{aligned}
    \tilde{Y}_{is'} &:=
    \frac{\delta_{s'+1} p_{s'+1}(\overline{\bm
      W}_{is'})m_{s'+1 \mid s }(\overline{\bm{W}}_{is'}, W_{i,s'+1} = 1, \ \{\boldf^{[s]}(\bR_{is''}, s'')\}_{s''=1}^{s'+1}; \bm\delta_{(s'+2):s}) }{\delta_{s'+1}  p_{s'+1}(\overline{\bm W}_{is'}) + 1 - p_{s'+1}(\overline{\bm
      W}_{is'})}\\
      &\quad \hspace{.5in} + \frac{\{1 -
      p_{s'+1}(\overline{\bm W}_{is'})\} m_{s'+1 \mid s }(\overline{\bm{W}}_{is'}, W_{i,s'+1} = 0, \ \{\boldf^{[s]}(\bR_{is''}, s'')\}_{s''=1}^{s'+1}; \bm\delta_{(s'+2):s})
      }{\delta_{s'+1}  p_{s'+1}(\overline{\bm W}_{is'}) + 1 - p_{s'+1}(\overline{\bm
      W}_{is'})}. \label{pseudo_outcome}
    \end{aligned}
\end{equation}
Intuitively, $\tilde Y_{is'}$ represents the expected outcome at the target segment $s$ when all future treatments from segment $s'+1$ onward are generated according to the stochastic intervention. Therefore, to encode the corresponding deconfounder requirement, we fit the same neural network architecture with the pseudo-outcome replacing the observed outcome. That is, we minimize
\begin{align*}
    \{\hat\blambda_{s'}, \hat\btheta_{s'}\} = \argmin_{\btheta_{s'}, \blambda_{s'}} \; \frac{1}{N_s} \ \sum_{i = 1}^N \ \mathbbm{1}\{S_i \geq s\} \left\{ \tilde Y_{is'} - m_{s'\mid s}(\overline{\bm{W}}_{is''}, \{\boldf^{(s)}(\bR_{is''},s'';
  \blambda_{s'})\}_{s''=1}^{s'}; \btheta_{s'}) \right\}^2.
\end{align*}

Given the above neural network architecture, we estimate the average
potential outcome using the semiparametric estimation method for the
longitudinal data. We first derive the influence function for the
following single outcome estimand at segment $s$:
\begin{align*}
\Psi_{s}(\bm{\delta}):= \E\biggl[ \int_{\mathcal{W}^{s}} Y_{is}(W_{i1} = w_1, \ldots, W_{is} = w_s)
\prod_{s'= 1}^{s} d Q_{s'}(w_{s'}; \overline{\bm{w}}_{s'-1}, \delta_{s'}) \mid S_i \geq s \biggr].
\end{align*}
We derive the influence function, denoted as
$\psi_s(\cD_{is}; \bm\delta, \bm\eta_s, \Psi_s)$ with observed data
$\cD_{is} = \{Y_{is}, \overline{\bW}_{is}, \overline{\bR}_{is}\}$ and
nuisance functions $\bm\eta_s$ in the case of single-outcome
estimand.  The influence function satisfies the asymptotic linear expansion,
\begin{align}
    \sqrt{N_s} \bigl( \widehat\Psi_s(\bm\delta) - \Psi_s(\bm\delta)\bigr) = \frac{1}{\sqrt{N_s}}\sum_{i = 1}^{N}\mathbf{1}\{S_i \geq s\} \psi_s(\cD_{is}; \bm\delta, \bm\eta_s, \Psi_s) + o_p(1) \label{asymp_linear}
\end{align}
where $N_s = \sum_{i = 1}^N \mathbf{1}\{S_i \geq s\}$, and its expectation is zero:
$\E[\psi_s(\cD_{is}; \bm\delta, \bm\eta_s, \Psi_s)] = 0$.

The asymptotic linear expansion allows us to derive the asymptotic
normality by applying the central limit theorem to
Equation~\eqref{asymp_linear}, whereas the mean-zero property enables
us to construct the estimating equation using the influence
function. Furthermore, the estimator constructed from the influence
function is robust to the regularization bias in the nuisance
function, the property known as \emph{Neyman orthogonality}, which
allows us to use flexible machine learning functions for estimation
(see \citealt{tsiatis_semiparametric_2006,
  chernozhukov_doubledebiased_2018, hines_demystifying_2022,
  kennedy_semiparametric_2023}). This feature is especially important
in our setting because video data are high-dimensional and
unstructured, making the standard parametric modeling assumptions
difficult to justify and necessitating the use of flexible machine
learning models like neural networks. Once we derive the influence function for the single outcome, we can then obtain the influence function for the temporally aggregated average potential outcome. See Appendix~\ref{section:theorem_if} for the exact expression of influence function for each case.

We use the cross-fitting procedure based on the derived influence function for the estimation \citep{laan_targeted_2006, chernozhukov_doubledebiased_2018}. In our setting, cross-fitting is important because the nuisance functions $\bm\eta$ are estimated using flexible machine-learning methods, including neural networks, which can accommodate the high-dimensional video representations but may also overfit the data. As a result, we avoid using the same observations both to train the nuisance functions and to evaluate the influence-function-based estimating equation. At the same time, simple sample splitting is inefficient because it uses only part of the sample for inference. Cross-fitting resolves this problem by rotating the held-out sample across folds, so that each observation is used for statistical inference exactly once while all observations contribute to nuisance-function estimation. See Appendix~\ref{section:crossfit_algorithm} for the entire estimation procedure.

\subsection{Estimating Average Potential Outcome Trajectory}

We finally consider the average potential outcome trajectory, which is defined as
\begin{align}
    \Vec{\bm\Psi}(\bm\delta) = (\Psi_1(\bm\delta), \cdots, \Psi_{s_{\max}}(\bm\delta)), \label{target_quantity2}
\end{align}
for $s = 1, \cdots, s_{\max}$. Unlike the average potential outcome
aggregated over segments considered in the previous section, this
quantity preserves the dynamic pattern of the causal effect by jointly
reporting the expected potential outcome for all segments. While the identification of this quantity 
is nearly identical to that in Theorem~\ref{theorem:identification} and we can use an analogous influence function
for each outcome, inference for this quantity
requires simultaneous uncertainty quantification over the entire
trajectory, rather than pointwise confidence intervals at each
interval.

We thus construct uniform confidence bands by applying a multiplier
bootstrap procedure to the vector of estimated influence functions
\citep{chernozhukov_gaussian_2013, kennedy_nonparametric_2019}. To
construct a $(1 - \alpha)$ uniform confidence band, we need to find a
critical value $\widehat c_{1-\alpha}(\bm\delta)$ that satisfies
\begin{align*}
    \P\left(
    \Psi_s(\bm\delta)
    \in
    \left[
    \widehat\Psi_s(\bm\delta)
    \pm
    \widehat c_{1-\alpha}(\bm\delta)
    \frac{\widehat\sigma_s(\bm\delta)}{\sqrt{N_s}}
    \right]
    \ \text{for all }\
    s=1,\ldots,s_{\max}
    \right)
    \to
    1-\alpha.
\end{align*}
We obtain such a critical value through the following multiplier bootstrap procedure:
\begin{enumerate}
\item Using the cross-fitting procedure as in the previous subsection, compute the estimated influence function and the variance estimate $$\widehat\sigma^2_s(\bm\delta) = 
    \frac{1}{N_s} 
    \sum_{k = 1}^K \sum_{I(i) = k}
    \mathbbm{1}\{S_i \geq s\}\psi_s
    \left(
    \cD_{is};
    \bm\delta,
    \hat{\bm\eta}^{(-k)},
    \widehat{\Psi}_s
    \right)^2$$
for each interval $s = 1, \ldots, s_{\max}$ and each observation $i = 1, \ldots, N$.

\item For each bootstrap iteration $b = 1, \ldots, B$, draw a set of independent Rademacher multipliers $\{\xi_i^{(b)}\}_{i=1}^N$ that are independent of the data and $\Pr(\xi_i = 1) = \Pr(\xi_i = -1) = 0.5$, and form the bootstrapped process
\begin{align*}
    \widehat{G}_s^{(b)}(\bm\delta) = \frac{1}{\sqrt{N_s}} \sum_{k = 1}^K \sum_{I(i) = k} \xi_i^{(b)} \frac{\mathbbm{1}\{S_i \geq s\} \ \psi_s
    \left(
    \cD_{is};
    \bm\delta,
    \hat{\bm\eta}^{(-k)},
    \widehat{\Psi}_s
    \right)}{\widehat\sigma_s(\bm\delta)}, \quad s = 1, \ldots, s_{\max}.
\end{align*}

\item Compute the maximum of the absolute bootstrapped process across all intervals,
\begin{align*}
    T^{(b)}(\bm\delta) = \max_{1 \le s \le s_{\max}} \bigl| \widehat{G}_s^{(b)}(\bm\delta) \bigr|.
\end{align*}

\item Set the critical value $\widehat c_{1-\alpha}(\bm\delta)$ to the empirical $(1 - \alpha)$ quantile of $\{T^{(b)}(\bm\delta)\}_{b=1}^B$.
\end{enumerate}
The resulting band attains the nominal level asymptotically under the same regularity condition for the asymptotic normality \citep{chernozhukov_gaussian_2013, kennedy_nonparametric_2019}. Importantly, this procedure does not require refitting nuisance functions for each bootstrap procedure; we only need to sample the Rademacher multipliers in each bootstrap sample.

\section{Video Causal Benchmark Dataset and Empirical Validation}\label{sec:mario}

In this section, we construct a novel video causal benchmark dataset with known ground truth and use it to empirically validate the proposed methodology. The main challenge is to create a video dataset, in which treatment and confounding features have complex and dynamic relationships with the outcome while the true causal effect remains known. We overcome this challenge with the help of Super Mario Bros.\texttrademark{}

\subsection{Mario Benchmark Environment}

\begin{figure}[t]
    \centering
    \includegraphics[width=0.5\linewidth]{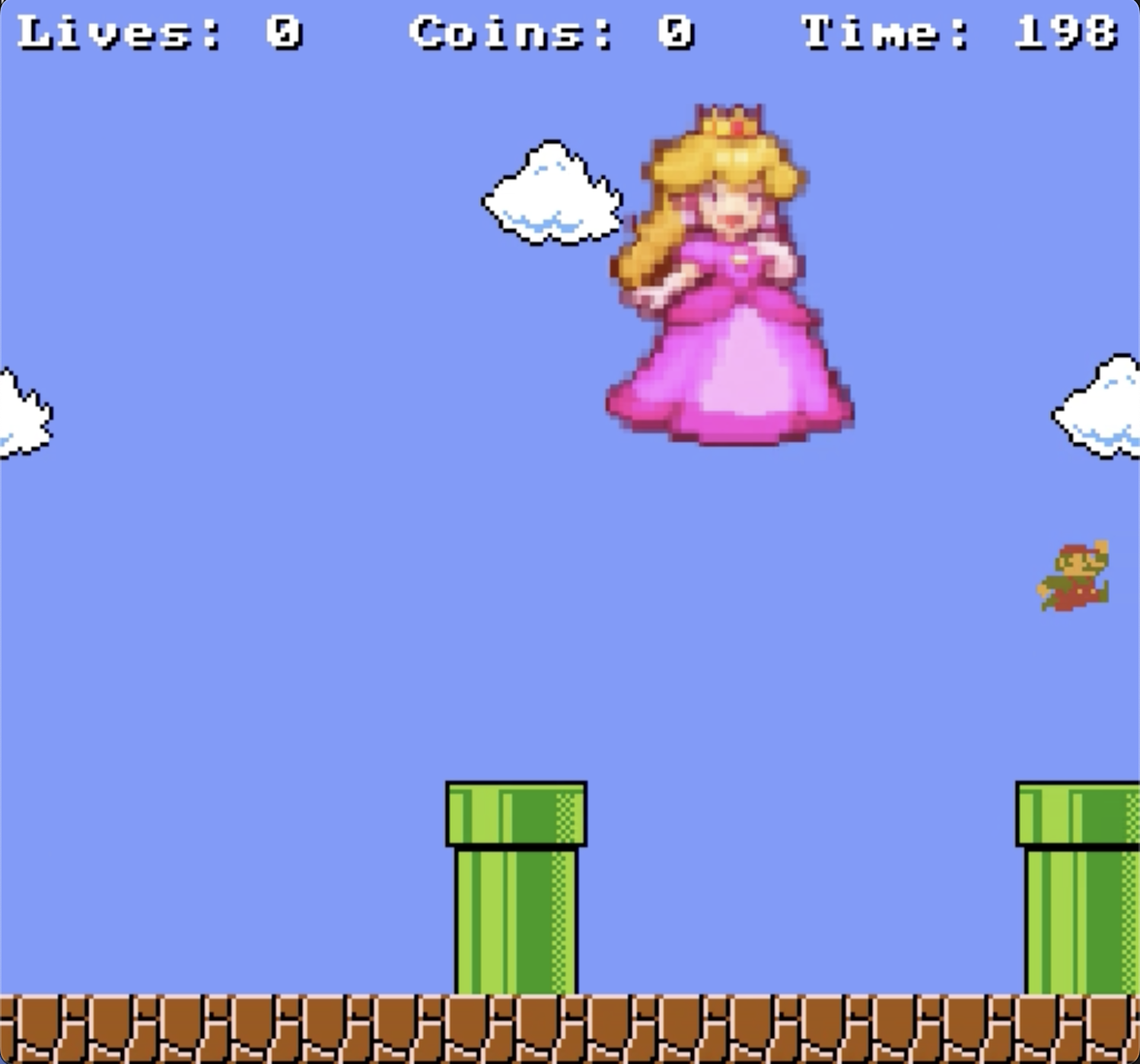}
    \caption{A Screenshot from the Mario environment. The treatment is the appearance of Princess Peach, while the outcome is the agent's number of jumps. Princess Peach has no effect on the Mario agent, but her appearance is correlated with the number of pipes, which induce jumping.}
    \label{mario_screenshot}
\end{figure}

We construct a simple Super Mario Bros.\texttrademark{} environment, in which the generation of time-varying treatment and confounding features is precisely controlled, while the outcome-generating process remains complex.  The environment is designed to mirror the structure of our campaign advertisement application, in which treatment features (e.g., a candidate's appearance) appear at different points in a video and may influence viewers' reactions. In our synthetic environment, the treatment is the appearance of Princess Peach floating at the top of the scene (Figure~\ref{mario_screenshot}) at various times throughout the video. Pipe obstacles serve as confounding features, with their appearances intentionally correlated with those of Princess Peach. The outcome is the \emph{reaction} of a fixed Glenn-Hartmann Mario AI agent \citep{karakovskiy2012mario}, measured by the number of jumps performed by Mario. Although the appearance of Princess Peach has no effect on the agent's behavior, pipe obstacles cause Mario to jump, creating a complex confounding structure while ensuring that the true causal effect of the treatment is zero.

In contrast to many interactive environments, the duration of a campaign advertisement is not affected by viewers’ reactions. To reflect this property, our Mario levels auto-scroll at a fixed rate, so that every video lasts exactly 26 seconds, corresponding to the average time required for a Glenn–Hartmann agent to complete a level. To generate each video, we partition the level $i$ into 2-second segments $s$ and first sample the pipe obstacles $U_{is}$ (confounders). 
%
% This segmentation mirrors the campaign-ad setting, where treatment features and confounding visual context both vary as the video unfolds, so adjustment strategies that ignore this dynamic structure can remain biased, as we show in Section~\ref{ssec:marioresults}.
%
In the initial video segment, we assign either 0 or 1 pipe with equal probability. To capture the fact that confounders evolve over the course of a video, the number of pipes in each subsequent segment follows a random walk: relative to the previous segment, the count decreases by one with probability $0.4$, remains unchanged with probability $0.4$, and increases by one with probability $0.2$. Any draw that would produce a negative number of pipes or exceed the maximum of five pipes that can fit onscreen in a segment is truncated to $0$ and $5$, respectively.

The appearance of Princess Peach in video segment $s$ of video $i$, denoted by $W_{is}$, serves as the treatment. Treatment assignment is correlated with the realized pipe count, so that segments containing more pipes are more likely to include Princess Peach. Specifically, the treatment assignment mechanism is given by
$$
W_{is} \sim \mathrm{Bernoulli}\left(\epsilon + (1-2\epsilon)\cdot \sigma(\alpha+\beta U_{is}) \right),
$$
where $\sigma$ is the sigmoid function, $\alpha$ controls the baseline log-odds of treatment assignment, and $\beta$ determines the strength of confounding. The parameter $\epsilon>0$ ensures that the treatment assignment probability is bounded away from 0 and 1, thereby satisfying the overlap (positivity) assumption in causal inference. Throughout our experiments, we set $\alpha=-1.5$, $\beta=1$, and $\epsilon=0.05$. This setup yields a strongly confounded treatment assignment mechanism; for example, the probability of Peach appearing in a segment is: $\Pr(W_{is}=1\mid U_{is}=0)=0.214$, $\Pr(W_{is}=1\mid U_{is}=1)=0.390$, $\Pr(W_{is}=1\mid U_{is}=2)=0.610$,
and $\Pr(W_{is}=1\mid U_{is}=3)=0.786$.

We repeat this procedure to generate $10{,}000$ levels. The Glenn-Hartmann Mario agent then plays each level, and we record the segment-level outcome $Y_{is}$ as the number of jumps made in video segment $s$. The resulting benchmark dataset captures a setting in which treatment features may appear multiple times throughout a video, confounding features evolve over time, and the ground-truth causal effect is known.

\subsection{Analysis Setup}\label{sec:mario_analysis}

We analyze the synthetic Mario data described above using the proposed methodology. Specifically, we first divide each video into one-second segments.
We choose 1 second, which is finer than level generation, because the rendered gameplay unfolds continuously and Mario can respond by jumping within the same two-second window.
We estimate the effect of Princess Peach’s appearance in segment $s-1$ on Mario’s counterfactual number of jumps in segment $s$. This one-segment lag ensures that the visual treatment precedes the behavioral reaction, since jumps are themselves part of the rendered video and may occur before or after Princess Peach appears within the same segment. %{\bf KI: I don't understand this sentence.  Please clarify $\rightarrow$ Done.}.  
For all methods, we restrict the confounding history to the previous five one-second segments, which is reasonable in this environment because Mario's jumps are driven primarily by recent screen content. This history requirement determines the analysis window: we exclude the first five outcome segments because they do not have a complete five-segment lagged history. %{\bf KI: Are you talking about our Markov assumption? This sentence is unclear. If that's the case, then we should bring out the last sentence of this paragraph earlier. $\rightarrow$  Done. Hope it clarifies. }. 
The last four seconds are excluded because pipes no longer appear near the end of the level due to the endgame flag and Mario rarely jumps. 

To implement the dynamic GPI methodology, we first regenerate the video content using the NVIDIA Cosmos tokenizer \citep{agarwal2025cosmos} and extract their internal representations. These representations contain the information used by those models to reproduce the generated videos. Because their dimensionality is substantially larger than the sample size, we follow \cite{imai2026causal} and apply a pooling strategy to obtain lower-dimensional representations.
Specifically, we average the internal representations over the temporal dimension within each segment. This pooling strategy is reasonable provided that the segments are sufficiently short and exhibit limited visual variation over time. After pooling, the dimensionality of the resulting representations is $(16, 40, 60)$ for video, where 16 denotes the number of channels, 40 the height, and 60 the width. For the video representations, we preserve their tensor structure rather than flattening them, using the original channel-by-height-by-width array as input to the convolutional encoder described below.

We then apply the proposed dynamic GPI methodology to the regenerated video frames, using 2-fold cross-fitting. In each fold, we first estimate the neural networks shown in Figure \ref{architecture}. To encode the video tensor, we use a three-dimensional convolutional encoder with channel sizes $(8,16,32)$ followed by a 256-dimensional output layer, which allows us to capture the structure of the unstructured input. We then estimate the shared representation for the deconfounder and the outcome head using feedforward networks with hidden-layer dimensions $(256,32)$ and $(64)$, respectively. All networks are trained with Adam optimizer with weight decay using a learning rate of $3.912\times 10^{-4}$, a batch size of 32, no dropout, and early stopping with a patience of 5. These hyperparameters were selected using Optuna with 100 trials over the segments included in the analysis and then fixed across all cross-fitting folds \citep{akiba_optuna_2019}. After estimating the deconfounder, we recursively fit the outcome regression and compute the expected counterfactual outcome under two incremental parameters, $\delta=0.5$ and $\delta=5.0$. We report their difference, $\widehat\Psi_s(5.0)-\widehat\Psi_s(0.5)$, which compares Mario's counterfactual number of jumps under a higher versus lower probability of Princess Peach appearing at time $s-1$. Standard errors are computed using the influence function in Equation~\eqref{if_formula}.

We compare our proposed methodology against two benchmark estimators. The first is an unadjusted marginal structural model (MSM) that uses the same dynamic stochastic intervention but does not control for any confounding features. Because the appearance of Princess Peach is correlated with the appearance of pipes, this estimator is expected to exhibit substantial bias. The second is a pipe-count-adjusted MSM that adjusts for the average number of pipes appearing within each segment. Although pipe appearances are the only source of confounding by construction, this adjustment relies on a coarse segment-level summary and therefore cannot account for more fine-grained confounding information, such as the spatial location of pipes or their timing within the segment. For both benchmarks, we use the same stochastic interventions and report the difference $\widehat\Psi_s(5.0)-\widehat\Psi_s(0.5)$. As all three methods target the same dynamic stochastic intervention, the only difference across methods is how they adjust for confounding.

\subsection{Results}\label{ssec:marioresults}

Overall, we find that methods that ignore complex video dynamics, such as the no-adjustment benchmark, or that adjust only for a coarse hand-coded proxy such as average pipe counts, would lead researchers to falsely conclude that Princess Peach’s appearance increases agent reactions (jumping), whereas Dynamic GPI recovers the known null effect by adjusting for richer, temporally evolving visual confounding.

\begin{figure}[t]
    \centering
    \includegraphics[width=1.0\linewidth]{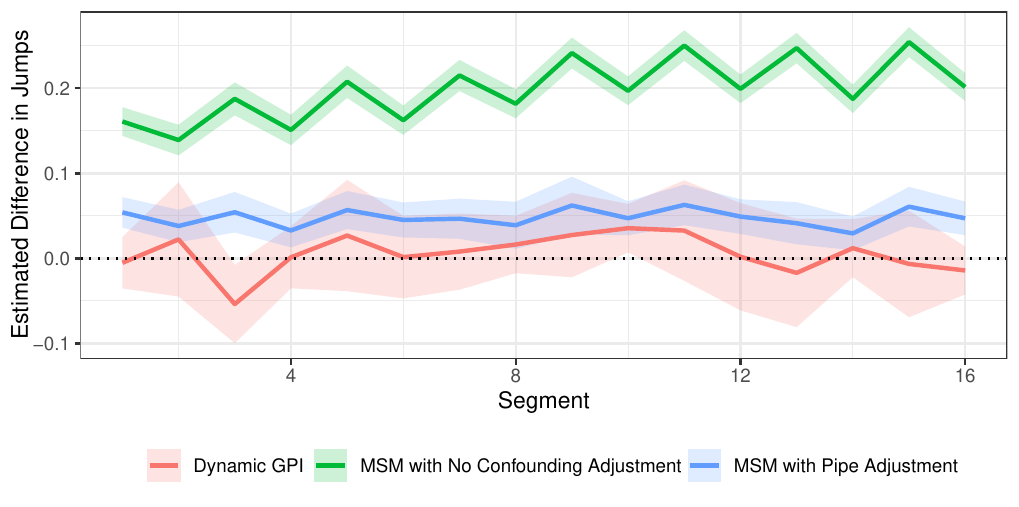}
    \caption{Estimated difference in the average number of jumps under two stochastic interventions: $\delta=5.0$, which increases the probability that Princess Peach appears, and $\delta=0.5$, which decreases this probability. Estimates are shown for each analyzed video segment. Segment 1 corresponds to seconds 5--6 of the original video. The figure compares the proposed dynamic GPI estimator, a marginal structural model (MSM) without confounding adjustment, and an MSM that adjusts for the average number of pipes within each segment. By construction, Princess Peach has no causal effect on Mario’s jumps, so the true effect is zero, shown by the black dotted line. Shaded regions denote 95\% uniform confidence bands computed using a multiplier bootstrap with 2,000 replications.}
    \label{mario_results}
\end{figure}

Specifically, Figure~\ref{mario_results} reports the estimated difference in the average number of jumps for each of the three estimators. 
The results show that the no-adjustment benchmark (green) is positive for all segments, and its confidence intervals do not include the ground truth of zero. This bias is due to confounding: the number of pipes is positively correlated with the appearance of Princess Peach, so the positive sign is expected. 

Adjusting for the average number of pipes within each segment (blue) removes most of the confounding bias, but the estimates remain slightly above the ground truth, and the confidence intervals do not include zero in any segment. This is because the pipe count discards important information about the confounders, such as the location and timing of the pipes. 

In contrast, the proposed Dynamic GPI methodology (red) generally includes the ground truth within its confidence intervals, and its point estimates are closest to zero. This shows that directly adjusting for video features enables researchers to account for richer aspects of confounding.

These results also suggest that common practices for analyzing unstructured data can be misleading. When analyzing texts and images, researchers often create a small number of hand-crafted features and adjust only for those low-dimensional summaries. However, there is no guarantee that such constructed variables alone capture all relevant aspects of confounding. This problem is especially severe for video data because videos are dynamic but the low-dimensional summaries often ignore temporal information. In contrast, the proposed Dynamic GPI methodology does not require researchers to prespecify which low-dimensional summaries are sufficient. Instead, it learns low-dimensional confounding features directly from the data.

\section{Empirical Analysis}\label{sec:application}

We apply the proposed methodology to the presidential campaign
advertisement data described in
Section~\ref{sec::example}. Specifically, we first generate
time-stamped transcripts using OpenAI’s \texttt{Whisper} \citep{radford2023robust}.
We then divide each video into consecutive 5-second segments. If a video's duration is not an exact multiple of 5 seconds, the final segment contains the remaining footage. For each video, we exclude the first and last segments from the analysis because treatment coding is less reliable at the boundaries of the video. For each segment, we regenerate the video content and transcript using the NVIDIA Cosmos tokenizer \citep{agarwal2025cosmos} and the instruction-fine-tuned LLaMA 3.1 model with 8 billion parameters \citep{grattafiori2024llama2}, respectively. Figure~\ref{reconstruction} presents an example of the original (left) and regenerated (right) video frames. The regenerated frames closely preserve the visual content of the original videos, indicating high-fidelity reconstruction.

\begin{figure}[t]
    \centering
    \includegraphics[width=1.0\linewidth]{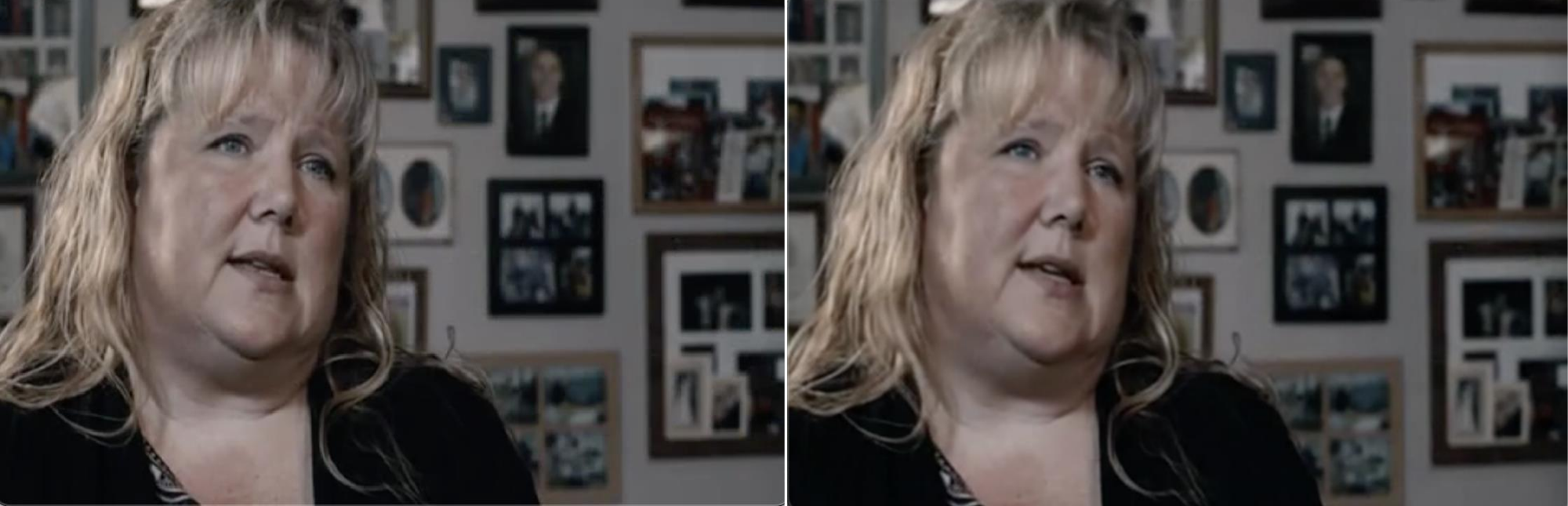}
    \caption{The original video frame (left) and the reconstructed video frame using the NVIDIA Cosmos tokenizer (right).  The reconstructed video frame preserves most visual features of the original frame.}
    \label{reconstruction}
\end{figure}

After regenerating the videos and transcripts, we extract their internal representations from the corresponding generative AI models. For the video modality, we follow the same representation-extraction and pooling procedure used in the Mario benchmark analysis described in Section~\ref{sec:mario_analysis}. For the transcripts, we use the hidden state corresponding to the final token of each regenerated text sequence, which contains the entire semantic information due to the autoregressive structure of the decoder-only model (e.g., LLaMA) and thus is a common choice in the literature (e.g., \citealt{neelakantan2022text, ma2024fine}). After pooling, the dimensionality of text representation is 4096. 

We then apply the proposed dynamic GPI methodology to the regenerated transcripts and video frames following the estimation procedure and training protocol used in the Mario benchmark analysis in Section~\ref{sec:mario}. We additionally adjust for respondent gender and partisanship to account for the stratified sampling design, and for the partisanship of the assigned video's candidate for the subsequent heterogeneity analysis. As before, we use 2-fold cross-fitting. In each fold, we first estimate the neural networks shown in Figure~\ref{architecture}. For text, we use a feedforward encoder with a 512-dimensional hidden layer and a 128-dimensional output representation. For video, we use a three-dimensional convolutional encoder with channel sizes $(16,32,64)$ followed by a 128-dimensional output layer. We then estimate the shared representation and the outcome head using feedforward networks with hidden-layer dimensions $(256,128)$ and $(64)$, respectively. As before, all networks are trained with Adam optimizer with weight decay using a learning rate of $7.304\times 10^{-5}$, a batch size of 4, dropout of 0.1, and early stopping with a patience of 5. As before, these parameters are chosen based on the same hyperparameter tuning procedure using Optuna with 100 trials. 

After estimating the deconfounder, we recursively fit the outcome regression and then compute the expected value of the counterfactual outcome under two incremental parameters, $\delta = 0.5$ and $\delta = 5$, along with its standard error using the influence function in Equation~\eqref{if_formula}.
We restrict the intervention analysis to segments through 25--30 seconds because the number of advertisements becomes small after 30 seconds, making estimates for later segments less stable. Finally, we use the estimated influence function to estimate the effect and its standard error separately for each candidate-party and respondent-party pair \citep{semenova2021debiased}.

\begin{figure}[t]
    \centering
    \includegraphics[width=1.0\linewidth]{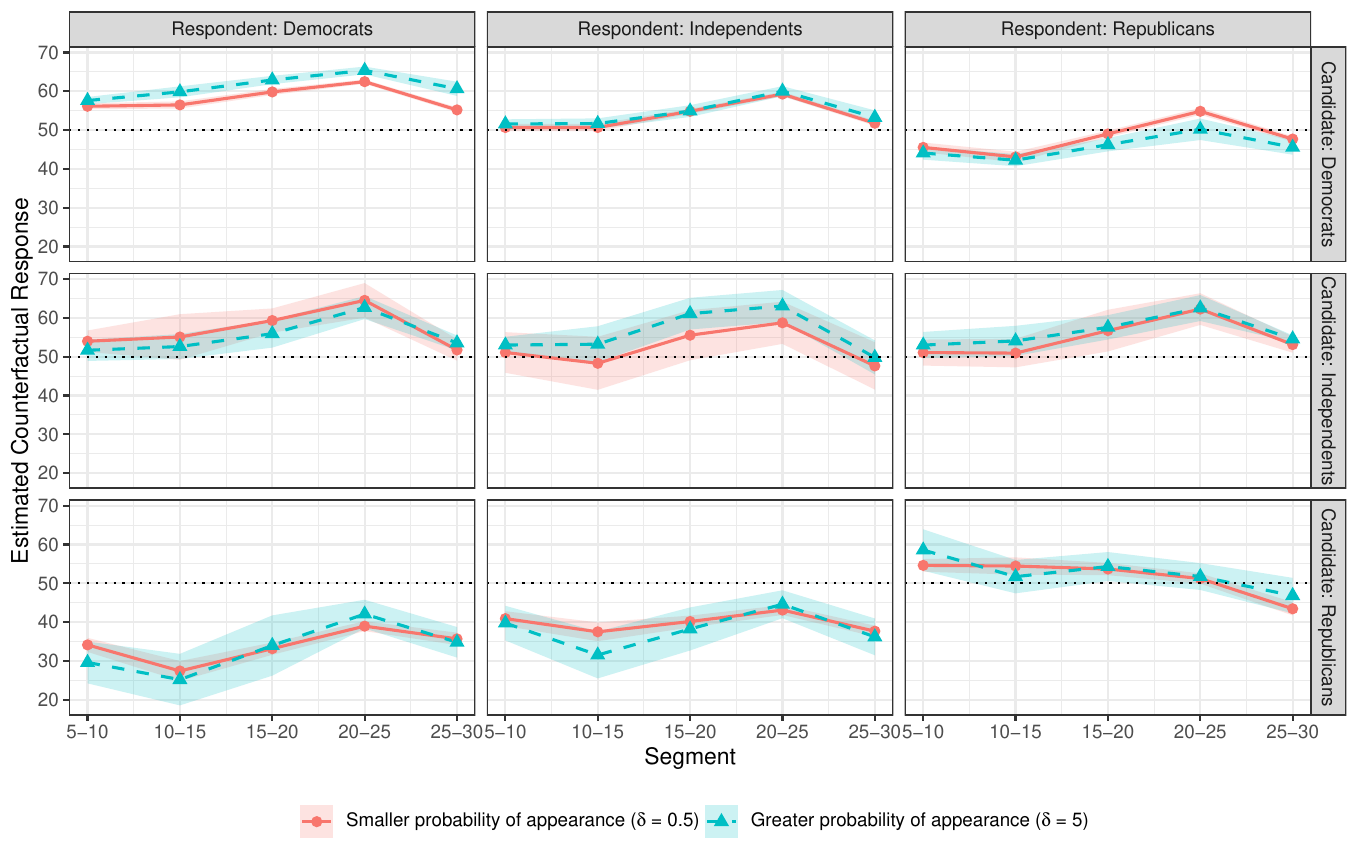}
    \caption{Estimated average potential outcome trajectory of
      candidate dial ratings for each respondent (columns) and candidate (rows) partisanship pair at each segment under two stochastic
      interventions: $\delta = 5$, which increases the probability
      that the candidate appears, and $\delta = 0.5$, which decreases the
      probability. The shaded region represents
      the 95\% uniform confidence bands computed using a multiplier
      bootstrap. The dotted line (50) is the initial value and
      represents neither cold nor warm.}
    \label{results_videos}
\end{figure}

Figure~\ref{results_videos} reports the estimated average potential outcome trajectories for each candidate (row) and respondent (column) partisanship pair. The results show that increasing the probability of Democratic candidate appearances improves average ratings among Democratic respondents (top left corner). This improvement occurs across segments and appears to be cumulative, as later segments show larger differences between the counterfactual trajectory under a higher probability of candidate appearance (blue) and the trajectory under a lower probability (red). 

This positive effect of partisan alignment is not observed among Republican and Independent respondents (middle and right columns). By contrast, increasing the probability of Democratic candidate appearances decreases average ratings among Republican respondents (top right corner), yielding a backlash effect. We do not observe an analogous backlash among Democratic respondents toward Republican candidates (bottom left corner), largely because Democratic respondents rate Republican candidate videos low, regardless of how frequently the candidates appear.

\section{Concluding Remarks}\label{sec:conclusion}

In this paper, we develop the first statistical methodology for causal inference with video features as treatments. Applying the GPI framework of \cite{imai2026causal, imai2026leveraging}, we first use a generative model to reproduce videos of interest and leverage its internal representations as learned, low-dimensional summaries for downstream causal inference. We establish the nonparametric identification and estimation of causal effects for treatment features that unfold over time while adjusting for latent confounding features in the video. To do so, we develop a longitudinal neural network architecture that learns a low-dimensional representation, which we call the deconfounder. Adjusting for the learned deconfounder yields a consistent and asymptotically normal estimator of the causal effect.

To validate our methodology, we construct a Super Mario Bros.\texttrademark{} video causal benchmark that is of independent interest for evaluating causal inference methods. The benchmark enables precise control over the generation of treatment and confounding features while using a Mario AI agent to generate complex real-time outcomes. We show that the proposed methodology accurately recovers the ground-truth causal effects.  
Finally, we apply the proposed methodology to television advertisements from the 2020 U.S. presidential campaign. We find that increasing the probability of a candidate appearing over time leads to more favorable average evaluations of the advertisement.

Future research can extend the proposed methodology in several directions. First, our current implementation relies on visual content and transcripts rather than a fully audiovisual representation of videos. In our application, we incorporate spoken content through transcripts instead of modeling the raw acoustic signal directly. This choice reflects a practical limitation of existing open-source generative video models: although they reproduce visual content with high fidelity, they cannot yet faithfully regenerate audio. As generative audio and audiovisual models continue to mature, we believe our framework can be extended to extract internal representations directly from the acoustic signal, allowing the deconfounder to capture confounding features that transcripts alone cannot recover. 

Second, the current framework requires researchers to segment videos into discrete components before estimation. Although these segments can in principle be made arbitrarily short, finer segmentation substantially increases the number of time intervals and can make estimation computationally expensive or infeasible. Future work should therefore develop more efficient and stable learning algorithms that remain scalable when the number of video segments is large.

\newpage 
\bibliography{my, references}

\newpage

\appendix
\setcounter{equation}{0}
\setcounter{figure}{0}
\setcounter{table}{0}
\setcounter{section}{0}
\renewcommand {\theequation} {S\arabic{equation}}
\renewcommand {\thefigure} {S\arabic{figure}}
\renewcommand {\thetable} {S\arabic{table}}
\renewcommand {\thesection} {S\arabic{section}}

\begin{center}
  {\LARGE \bf Appendix} 
\end{center}

\section{Details of Data and Measurement}
\label{app:surveys}

%\subsection{Measuring Treatment and Outcome Variables}

\begin{figure}[h]
  \centering
  \begin{subfigure}[t]{0.66\textwidth}
    \centering
    \includegraphics[width=0.4575\linewidth]{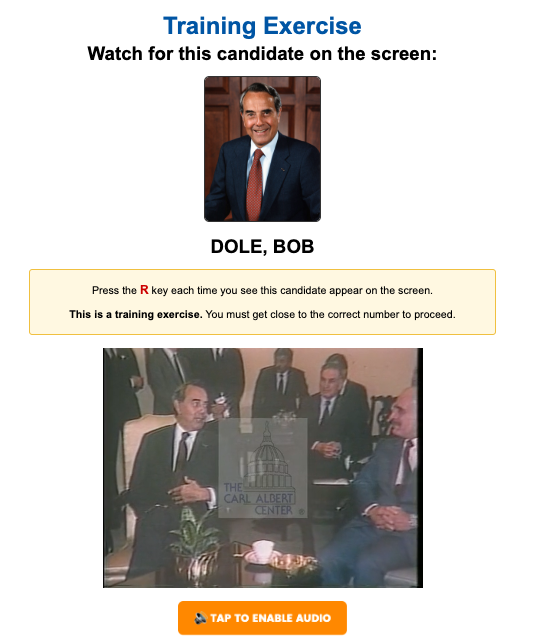} 
    \includegraphics[width=0.4425\linewidth]{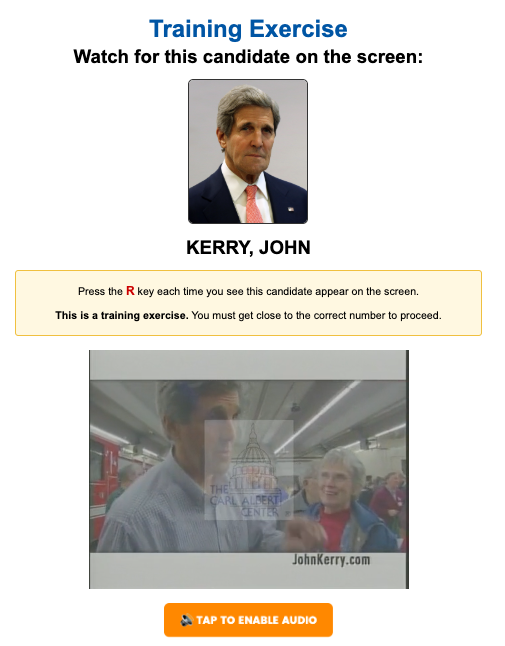}
    \caption{Training tasks for candidate identification}
    \label{fig:prolific1a}
  \end{subfigure}\hfill
   \begin{subfigure}[t]{0.33\textwidth}
    \centering
    \includegraphics[width=\linewidth]{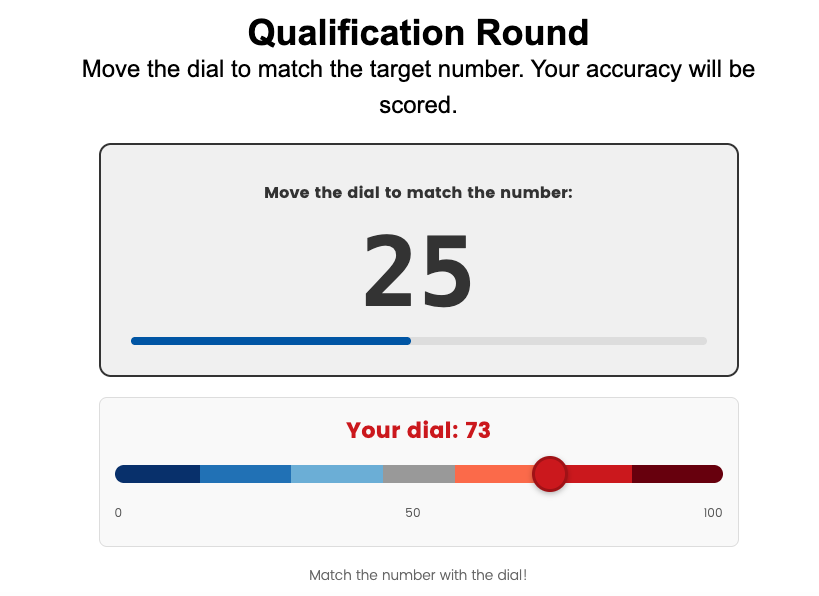}
    \caption{Training task for candidate rating}
    \label{fig:prolific2a}
      \end{subfigure}\hfill 
  \caption{Screening of Respondents for the Candidate Identification and Rating Tasks. Figure~\ref{fig:prolific1a} shows the screenshots of two historical training advertisements used to screen respondents for the candidate identification task.   Figure~\ref{fig:prolific2a} shows the training task used to train respondents for the candidate rating task.}
  \label{fig:prolific1}
\end{figure}

We screened respondents prior to the candidate identification task
using questions based on the two historical advertisements shown in
Figure~\ref{fig:prolific1a}. Specifically, respondents were randomly
assigned to one of the two advertisements and given up to three
attempts to identify the correct number of candidate appearances,
using researcher-established ground truth. Respondents who failed to
answer correctly within three attempts were exited from the survey. An
attempt was scored as correct if the number of appearances a
respondent identified was within one of the ground-truth count, which
was 5 for the Dole advertisement and 7 for the Kerry
advertisement. The timing of the keypresses themselves was not
evaluated. In total, 432 respondents began the survey, of whom 8 did
not successfully complete the screening exercise, resulting in a final
sample of 424 participants.

Using Qualtrics’ ``Evenly Present Elements'' feature, we sought to
assign 10 advertisements to each respondent while balancing the number
of evaluations across advertisements. In practice,
    respondents coded an average of 9.97 advertisements each with the
    standard deviation of 0.26. Each advertisement was coded by an
    average of 4.98 unique respondents (SD = 0.55). 
    For each
    advertisement, we computed the share of coder pairs who agreed on
    whether the candidate appeared in each five-second interval,
    averaged across all of the advertisement's intervals. Across the
    849 advertisements, the mean pairwise agreement rate was 0.91 (SD
    = 0.07). % and Krippendorff's $\alpha$ was 0.79, both of which
            % comport with conventional standards of reliability
            % \citep{krippendorff2004reliability}.
    We found 42 ads in which the candidate never appeared in any
    segment. We manually reviewed all 42: in 33, the candidate indeed
    never appears on the screen, and in the remaining 9, the candidate
    only appears in the closing shot as a still image. We also found
    24 ads, in which the candidate appeared only in the final interval
    and was otherwise never on screen. We manually reviewed all 24: in
    every instance, the candidate indeed only appears in that closing
    interval. In sum, we find that our measurement of the treatment
    variable is accurate.

We similarly screened respondents prior to the candidate rating task. Specifically, respondents first completed a training exercise in which a randomly generated number appeared on the screen, and they were instructed to move the rating dial to match that number (see Figure~\ref{fig:prolific2a}). The exercise consisted of a practice round followed by a qualifying round. Respondents who failed to complete the qualifying round within two attempts were exited from the survey. Successful completion required meeting two criteria: (a) a mean absolute error of no more than 15 points on the 0–100 scale between the dial position and the target value, and (b) moving the dial no more than one fewer time than the number of times the target changed on the screen, ensuring that respondents actively tracked each new target. After passing the qualifying round, respondents completed a practice rating using a randomly assigned advertisement before beginning the main task. If this advertisement was subsequently assigned during the main task, the corresponding rating was excluded from the analysis. This occurred only 25 times. Consequently, no practice ratings were included in the final dataset.

We drew 185 respondents from each of six cells crossed by partisanship (Democrat, Republican, and Independent) and gender (male or female), as indicated by Prolific. In total, 1,116 participants took the survey, 7 of whom did not successfully complete the training exercise, resulting in a total of 1,109 final participants (Democratic Female: 185; Democratic Male: 184; Republican Female: 184; Republican Male: 185; Independent Male: 186; Independent Female: 185). One respondent completed the survey a second time, with an independent random draw of advertisements, and is counted only once, leaving 1,108 unique respondents.
Using the ``Evenly Present Elements'' option in Qualtrics, we aimed to randomly assign 10 advertisements to each respondent, targeting 2 respondents per cell for each advertisement. 

\section{Additional Theoretical Results}

\subsection{Influence Function}\label{section:theorem_if}
\begin{theorem}[Influence Function for Single Outcome]\label{theorem_if}
Denote the observed data by $\cD := \{\cD_i\}_{i=1}^N$, where
$\cD_i := \{S_i,(Y_{it},W_{it},\bR_{it})_{t=1}^{S_i}\}$, and let
$\rho_s := \P(S_i\geq s)>0$. Denote the nuisance functions by
$\bm\eta_s := (\rho_s,\{\pi_{\ell\mid s},m_{\ell\mid s},
\tilde m_{\ell\mid s},p_\ell,\boldf^{(s)[\ell]}\}_{\ell=1}^s)$.
Suppose the deconfounders satisfy
bounded relative overlap. The influence function for
the target parameter $\Psi_s(\bm\delta)$ is given by
\begin{align}
 \psi_s(\cD_i;\bm\delta,\bm\eta_s,\Psi_s)
 = \frac{\varphi_s(\cD_i;\bm\delta,\bm\eta_s)
 -\mathbbm{1}\{S_i\geq s\}\Psi_s(\bm\delta)}{\rho_s}.
 \label{if_formula}
\end{align}
Here, $\varphi_s(\cD_i;\bm\delta,\bm\eta_s)$ is the uncentered
influence function for $\rho_s\Psi_s(\bm\delta)$, given by
\begin{equation}
\begin{aligned}
 &\varphi_s(\cD_i;\bm\delta,\bm\eta_s)\\
 :=\ &\mathbbm{1}\{S_i\geq s\}
 \sum_{\ell=1}^{s}\left\{\prod_{r=1}^{\ell-1}\omega_{ir\mid s}\right\}\\
 &\quad\times
 \frac{\begin{array}{@{}l@{}}
 \delta_\ell p_\ell(\overline{\bm W}_{i,\ell-1})
 m_{\ell\mid s}(\overline{\bW}_{i,\ell-1},1,
 \overline{\bm F}_{i\ell\mid s};\bm\delta_{(\ell+1):s})
 \left\{1-\frac{W_{i\ell}}{\pi_{\ell\mid s}(\overline{\bm H}_{i\ell\mid s})}\right\}
 \\[2pt]
 \quad+\{1-p_\ell(\overline{\bm W}_{i,\ell-1})\}
 m_{\ell\mid s}(\overline{\bW}_{i,\ell-1},0,
 \overline{\bm F}_{i\ell\mid s};\bm\delta_{(\ell+1):s})
 \left\{1-\frac{1-W_{i\ell}}{1-\pi_{\ell\mid s}(\overline{\bm H}_{i\ell\mid s})}\right\}
 \end{array}}
 {\delta_\ell p_\ell(\overline{\bm W}_{i,\ell-1})
 +1-p_\ell(\overline{\bm W}_{i,\ell-1})}\\
 &\quad+\sum_{\ell=1}^{s}\mathbbm{1}\{S_i\geq\ell\}
 \frac{\delta_\ell\{W_{i\ell}-p_\ell(\overline{\bm W}_{i,\ell-1})\}}
 {\{\delta_\ell p_\ell(\overline{\bm W}_{i,\ell-1})
 +1-p_\ell(\overline{\bm W}_{i,\ell-1})\}^2}\\
 &\qquad\times\left\{
 \tilde m_{\ell\mid s}(\overline{\bm W}_{i,\ell-1},1;\bm\delta_{-\ell})
 -\tilde m_{\ell\mid s}(\overline{\bm W}_{i,\ell-1},0;\bm\delta_{-\ell})
 \right\}\\
 &\quad+\mathbbm{1}\{S_i\geq s\}
 \left\{\prod_{\ell=1}^{s}\omega_{i\ell\mid s}\right\}Y_{is},
\end{aligned}\label{uncentered_if_formula}
\end{equation}
where for $\ell=1,\ldots,s$, we define
\begin{align*}
 \overline{\bm F}_{i\ell\mid s}
 &:=\left\{\boldf^{(s)[\ell]}(\bR_{ir},r)\right\}_{r=1}^{\ell},\\
 \overline{\bm H}_{i\ell\mid s}
 &:=\left\{\overline{\bW}_{i,\ell-1},
 \left\{\boldf^{(s)[t]}(\bR_{ir},r):1\leq r\leq t\leq\ell\right\}\right\},\\
 p_\ell(\overline{\bW}_{i,\ell-1})
 &:=\P(W_{i\ell}=1\mid\overline{\bW}_{i,\ell-1},S_i\geq\ell),\\
 \pi_{\ell\mid s}(\overline{\bm H}_{i\ell\mid s})
 &:=\P(W_{i\ell}=1\mid\overline{\bm H}_{i\ell\mid s},S_i\geq s),\\
 \omega_{i\ell\mid s}
 &:=\omega_{\ell\mid s}(\overline{\bm H}_{i\ell\mid s},W_{i\ell};
 \delta_\ell,p_\ell,\pi_{\ell\mid s})\\
 &:=\frac{
 \delta_\ell W_{i\ell}
 \frac{p_\ell(\overline{\bm W}_{i,\ell-1})}
 {\pi_{\ell\mid s}(\overline{\bm H}_{i\ell\mid s})}
 +(1-W_{i\ell})
 \frac{1-p_\ell(\overline{\bm W}_{i,\ell-1})}
 {1-\pi_{\ell\mid s}(\overline{\bm H}_{i\ell\mid s})}}
 {\delta_\ell p_\ell(\overline{\bm W}_{i,\ell-1})
 +1-p_\ell(\overline{\bm W}_{i,\ell-1})}.
\end{align*}
The superscript $[t]$ indicates the level of the outcome regression,
with $\boldf^{(s)[t]}(\bR_{ir},r):=\boldf_r^{(s)[t]}(\bR_{ir})$.
We also define
\begin{equation}
\begin{aligned}
 &\tilde m_{\ell\mid s}(\overline{\bm W}_{i,\ell-1},w;
 \bm\delta_{-\ell})\\
 :=\ &\E\Biggl[
 \mathbbm{1}\{S_i\geq s\}
 \left\{\prod_{r=1}^{\ell-1}\omega_{ir\mid s}\right\}
 m_{\ell\mid s}(\overline{\bW}_{i,\ell-1},w,
 \overline{\bm F}_{i\ell\mid s};\bm\delta_{(\ell+1):s})
 \ \bigg|\ \overline{\bm W}_{i,\ell-1},S_i\geq\ell
 \Biggr],
\end{aligned}\label{eq:video_marginalized_regression}
\end{equation}
for $w\in\{0,1\}$, where
$\bm\delta_{-\ell}:=(\delta_1,\ldots,\delta_{\ell-1},
\delta_{\ell+1},\ldots,\delta_{s_{\max}})$ and
$\bm\delta_{a:b}=\emptyset$ for $a>b$. 

We also recursively define
\begin{equation}
\begin{aligned}
 &m_{s\mid s}(\overline{\bW}_{is},\overline{\bm F}_{is\mid s};\emptyset)
 :=\mu_s(\overline{\bW}_{is},\overline{\bm F}_{is\mid s}),\\
 &m_{\ell\mid s}(\overline{\bW}_{i\ell},\overline{\bm F}_{i\ell\mid s};
 \bm\delta_{(\ell+1):s})\\
 :=\ &\E\Biggl[
 \frac{\delta_{\ell+1}p_{\ell+1}(\overline{\bm W}_{i\ell})
 m_{\ell+1\mid s}(\overline{\bW}_{i\ell},1,
 \overline{\bm F}_{i,\ell+1\mid s};\bm\delta_{(\ell+2):s})+\{1-p_{\ell+1}(\overline{\bm W}_{i\ell})\}
 m_{\ell+1\mid s}(\overline{\bW}_{i\ell},0,
 \overline{\bm F}_{i,\ell+1\mid s};\bm\delta_{(\ell+2):s})
 }
 {\delta_{\ell+1}p_{\ell+1}(\overline{\bm W}_{i\ell})
 +1-p_{\ell+1}(\overline{\bm W}_{i\ell})}
 \\[-2pt]
 &\hspace{2in}\bigg|\ \overline{\bW}_{i\ell},
 \overline{\bm F}_{i\ell\mid s},S_i\geq s\Biggr],
\end{aligned}\label{outcome_model}
\end{equation}
for $\ell=0,\ldots,s-1$, with
$\overline{\bW}_{i0}=\overline{\bm F}_{i0\mid s}=\emptyset$.
\end{theorem}
The proof is omitted, as it follows by applying the chain-rule
argument of Theorem~2 of \cite{nakamura2026genai} to
$\rho_s\Psi_s(\bm\delta)$ and then the ratio rule. The conditioning
event in Equation~\eqref{eq:video_marginalized_regression} accounts
for the fact that $p_\ell$ is defined among respondents with
$S_i\geq\ell$, whereas $\Psi_s$ is defined among those with $S_i\geq s$.

\begin{proposition}[Influence Function for Temporally Aggregated Average Potential Outcome]
\label{prop:if_agg}
Denote the uncentered influence function for $\rho_s\Psi_s(\bm\delta)$
by
\begin{align*}
 \varphi_s(\cD_i;\bm\delta,\bm\eta_s)
 :=\rho_s\psi_s(\cD_i;\bm\delta,\bm\eta_s,\Psi_s)
 +\mathbbm{1}\{S_i\geq s\}\Psi_s(\bm\delta),
\end{align*}
and the nuisance functions by
$\bm\eta:=\{\bm\eta_s\}_{s=1}^{s_{\max}}$.
The influence function for the target parameter $\bar\Psi(\bm\delta)$
is given by
\begin{align*}
 \bar\psi(\cD_i;\bm\delta,\bm\eta,\bar\Psi)
 :=\sum_{s=1}^{s_{\max}}\varphi_s(\cD_i;\bm\delta,\bm\eta_s)
 -\bar\Psi(\bm\delta).
\end{align*}
\end{proposition}
The proof is in Appendix~\ref{proof:if_agg}.

\subsection{Estimation Algorithm for Temporally Aggregated Average Potential Outcome}
\label{section:crossfit_algorithm}

\begin{enumerate}
\item Set the incremental parameter $\bm\delta\in(0,\infty)^{s_{\max}}$.

\item Randomly partition the data into $K$ folds of equal size where
 the size of each fold is $n=N/K$. The observation index is denoted
 by $I(i)\in\{1,\ldots,K\}$, where $I(i)=k$ implies that the $i$th
 observation belongs to the $k$th fold. Let
 $\mathcal I_k:=\{i:I(i)=k\}$, keeping all segments of each respondent
 in the same fold.

\item For each fold $k=1,\ldots,K$, use the observations with
 $I(i)\neq k$ as training data:
\begin{enumerate}
\item For each $\ell=1,\ldots,s_{\max}$, estimate the conditional
 treatment probability given the past treatment sequence, which is
 denoted by
\begin{align*}
 \hat p_\ell^{(-k)}(\overline{\bW}_{i,\ell-1})
 :=\widehat{\P}^{(-k)}\left(W_{i\ell}=1\mid
 \overline{\bW}_{i,\ell-1},S_i\geq\ell\right).
\end{align*}
These estimates are shared across target outcomes $s\geq\ell$.

\item For each target outcome $s=1,\ldots,s_{\max}$, set
 $\widetilde Y_{is\mid s}^{(-k)}:=Y_{is}$ and proceed backward through
 $\ell=s,s-1,\ldots,1$. At each level, simultaneously obtain the
 estimated deconfounder and estimated outcome regression, denoted by
 $\hat{\boldf}^{(s)[\ell],(-k)}$ and
 $\hat m_{\ell\mid s}^{(-k)}$, by solving
\begin{align*}
 &\{\hat\blambda_{\ell\mid s}^{(-k)},
 \hat\btheta_{\ell\mid s}^{(-k)}\}\\
 &\quad:=\argmin_{\blambda,\btheta}\;
 \frac{1}{N_s^{(-k)}}
 \sum_{\substack{I(i)\neq k\\S_i\geq s}}
 \Biggl\{\widetilde Y_{i\ell\mid s}^{(-k)}
 -m_{\ell\mid s}\left(
 \overline{\bW}_{i\ell},
 \{\boldf^{(s)[\ell]}(\bR_{ir},r;\blambda)\}_{r=1}^{\ell};
 \bm\delta_{(\ell+1):s},\btheta\right)\Biggr\}^2,
\end{align*}
where $N_s^{(-k)}:=\sum_{I(i)\neq k}\mathbbm{1}\{S_i\geq s\}$ and
\begin{align*}
 \hat{\boldf}^{(s)[\ell],(-k)}(\bR_{ir},r)
 &:=\boldf^{(s)[\ell]}(\bR_{ir},r;
 \hat\blambda_{\ell\mid s}^{(-k)}),\\
 \widehat{\overline{\bm F}}_{i\ell\mid s}^{(-k)}
 &:=\{\hat{\boldf}^{(s)[\ell],(-k)}(\bR_{ir},r)\}_{r=1}^{\ell}.
\end{align*}
At $\ell=s$, this minimization coincides with
Equation~\eqref{eq:minimization} restricted to the training data, and
$\hat m_{s\mid s}^{(-k)}=\hat\mu_s^{(-k)}$.
For each training observation with $S_i\geq s$, compute
\begin{align*}
 &\widetilde Y_{i,\ell-1\mid s}^{(-k)}\\
 &\quad:=\frac{
 \delta_\ell\hat p_\ell^{(-k)}(\overline{\bm W}_{i,\ell-1})
 \hat m_{\ell\mid s}^{(-k)}\left(
 \overline{\bW}_{i,\ell-1},1,
 \widehat{\overline{\bm F}}_{i\ell\mid s}^{(-k)};
 \bm\delta_{(\ell+1):s}\right)}
 {\delta_\ell\hat p_\ell^{(-k)}(\overline{\bm W}_{i,\ell-1})
 +1-\hat p_\ell^{(-k)}(\overline{\bm W}_{i,\ell-1})}\\
 &\qquad+\frac{
 \{1-\hat p_\ell^{(-k)}(\overline{\bm W}_{i,\ell-1})\}
 \hat m_{\ell\mid s}^{(-k)}\left(
 \overline{\bW}_{i,\ell-1},0,
 \widehat{\overline{\bm F}}_{i\ell\mid s}^{(-k)};
 \bm\delta_{(\ell+1):s}\right)}
 {\delta_\ell\hat p_\ell^{(-k)}(\overline{\bm W}_{i,\ell-1})
 +1-\hat p_\ell^{(-k)}(\overline{\bm W}_{i,\ell-1})}.
\end{align*}
For $\ell>1$, use this pseudo-outcome at the next level of the
backward recursion.

\item For each target outcome $s=1,\ldots,s_{\max}$ and each
 $\ell=1,\ldots,s$, construct
\begin{align*}
 \widehat{\overline{\bm H}}_{i\ell\mid s}^{(-k)}
 :=\left\{\overline{\bW}_{i,\ell-1},
 \left\{\hat{\boldf}^{(s)[t],(-k)}(\bR_{ir},r):
 1\leq r\leq t\leq\ell\right\}\right\}.
\end{align*}
Using the training observations with $S_i\geq s$, estimate the
propensity score given the estimated deconfounder history, which is
denoted by
\begin{align*}
 \hat\pi_{\ell\mid s}^{(-k)}\left(
 \widehat{\overline{\bm H}}_{i\ell\mid s}^{(-k)}\right)
 :=\widehat{\P}^{(-k)}\left(W_{i\ell}=1\mid
 \widehat{\overline{\bm H}}_{i\ell\mid s}^{(-k)},S_i\geq s\right).
\end{align*}

\item For each training or held-out observation with $S_i\geq s$
 and each $\ell=1,\ldots,s$, compute
\begin{align*}
 \hat\omega_{i\ell\mid s}^{(-k)}
 &:=\frac{
 \delta_\ell W_{i\ell}
 \frac{\hat p_\ell^{(-k)}(\overline{\bm W}_{i,\ell-1})}
 {\hat\pi_{\ell\mid s}^{(-k)}(
 \widehat{\overline{\bm H}}_{i\ell\mid s}^{(-k)})}
 +(1-W_{i\ell})
 \frac{1-\hat p_\ell^{(-k)}(\overline{\bm W}_{i,\ell-1})}
 {1-\hat\pi_{\ell\mid s}^{(-k)}(
 \widehat{\overline{\bm H}}_{i\ell\mid s}^{(-k)})}}
 {\delta_\ell\hat p_\ell^{(-k)}(\overline{\bm W}_{i,\ell-1})
 +1-\hat p_\ell^{(-k)}(\overline{\bm W}_{i,\ell-1})},\\
 \tilde\omega_{i\ell\mid s}^{(-k)}(\bm\delta_{1:\ell})
 &:=\prod_{r=1}^{\ell}\hat\omega_{ir\mid s}^{(-k)},
 \qquad \tilde\omega_{i0\mid s}^{(-k)}:=1.
\end{align*}

\item For each target outcome $s=1,\ldots,s_{\max}$, each
 $\ell=1,\ldots,s$, and each $w\in\{0,1\}$, estimate
 $\widehat{\tilde m}_{\ell\mid s}^{(-k)}(
 \overline{\bm W}_{i,\ell-1},w;\bm\delta_{-\ell})$
 by the saturated regression of
\begin{align*}
 \mathbbm{1}\{S_i\geq s\}\,
 \tilde\omega_{i,\ell-1\mid s}^{(-k)}(\bm\delta_{1:(\ell-1)})
 \hat m_{\ell\mid s}^{(-k)}\left(
 \overline{\bW}_{i,\ell-1},w,
 \widehat{\overline{\bm F}}_{i\ell\mid s}^{(-k)};
 \bm\delta_{(\ell+1):s}\right)
\end{align*}
on $\overline{\bW}_{i,\ell-1}$ among training observations with
$S_i\geq\ell$. Set the regression response to zero for observations
with $\ell\leq S_i<s$. Here, $w$ is a fixed argument of the outcome
model, rather than a restriction on the observed treatment.
\end{enumerate}

\item Let
\begin{align*}
 \hat\rho_s
 &:=\frac{1}{N}\sum_{i=1}^N\mathbbm{1}\{S_i\geq s\},\\
 \hat{\bm\eta}_s^{(-k)}
 &:=\left(\hat\rho_s,
 \left\{\hat\pi_{\ell\mid s}^{(-k)},
 \hat m_{\ell\mid s}^{(-k)},
 \widehat{\tilde m}_{\ell\mid s}^{(-k)},
 \hat p_\ell^{(-k)},\hat{\boldf}^{(s)[\ell],(-k)}
 \right\}_{\ell=1}^{s}\right),\\
 \hat{\bm\eta}^{(-k)}
 &:=\left\{\hat{\bm\eta}_s^{(-k)}\right\}_{s=1}^{s_{\max}}.
\end{align*}
For each $i\in\mathcal I_k$, evaluate
$\varphi_s(\cD_i;\bm\delta,\hat{\bm\eta}_s^{(-k)})$ for every
$s=1,\ldots,s_{\max}$. For $s>S_i$, only the terms corresponding to
estimation of $p_\ell$ with $\ell\leq S_i$ can be nonzero.
Compute the estimator $\widehat{\bar\Psi}(\bm\delta)$ as a solution to
\begin{align*}
 \frac{1}{nK}\sum_{k=1}^K\sum_{I(i)=k}
 \bar\psi\left(\cD_i;\bm\delta,
 \hat{\bm\eta}^{(-k)},\widehat{\bar\Psi}\right)=0,
\end{align*}
and compute
\begin{align*}
 \hat{\bar\sigma}^2(\bm\delta)
 =\frac{1}{nK}\sum_{k=1}^K\sum_{I(i)=k}
 \bar\psi\left(\cD_i;\bm\delta,
 \hat{\bm\eta}^{(-k)},\widehat{\bar\Psi}\right)^2.
\end{align*}
\end{enumerate}

\section{Theoretical Proofs}

\subsection{Proof of Proposition~\ref{prop:if_agg}}\label{proof:if_agg}
\begin{proof}
First, notice that
 \begin{align*}
     \bar{\Psi}(\bm\delta) = \E\biggl[ \sum_{s = 1}^{s_{\max}} \mathbbm{1}\{S_i \geq s\} \Psi_s(\bm\delta) \biggr] = \sum_{s = 1}^{s_{\max}} \P(S_i \geq s) \Psi_s(\bm\delta).
\end{align*}
Now, from example 1 of \cite{hines_demystifying_2022}, the influence function for $\P(S_i \geq s)$ is given by
\begin{align*}
     \psi^{(S_i \geq s)}(\bm\delta) = \mathbbm{1}\{S_i \geq s\} - \P(S_i \geq s).
 \end{align*}
Therefore, by product rule, the influence function for $\bar{\Psi}(\bm\delta)$ is given by
\begin{align*}
     \bar{\psi}(\cD_{i}; \bm\delta, \bm\eta, \bar\Psi) &= \sum_{s=1}^{s_{\max}}
 \biggl\{ \Psi_s(\bm\delta) \psi^{(S_i \geq s)}(\bm\delta)  + \mathbbm{1}\{S_{i} \geq s\}
 \psi_s(\cD_{is}; \bm\delta, \bm\eta_s, \Psi_s) \biggr\}\\
 &= \sum_{s=1}^{s_{\max}}
 \biggl\{
 \{ \mathbbm{1}\{S_i \geq s\} - \P(S_i \geq s) \}\Psi_s(\bm\delta)
 +
 \mathbbm{1}\{S_i \geq s\}
 \{ \varphi_s(\cD_{is}; \bm\delta, \bm\eta_s) - \Psi_s(\bm\delta) \}
 \biggr\}\\
 &= \sum_{s=1}^{s_{\max}} \mathbbm{1}\{S_i \geq s\} \, \varphi_s(\cD_{is}; \bm\delta, \bm\eta_s) - \sum_{s=1}^{s_{\max}} \P(S_i \geq s) \Psi_s(\bm\delta)\\
 &= \sum_{s=1}^{s_{\max}} \mathbbm{1}\{S_i \geq s\} \, \varphi_s(\cD_{is}; \bm\delta, \bm\eta_s) - \bar{\Psi}(\bm\delta).
\end{align*}
\end{proof}

\begin{comment}
\subsection{Proof of Proposition~\ref{prop:if_agg}}\label{proof:if_agg}
\begin{proof}
First, notice that
\begin{align*}
    \bar{\Psi}(\bm\delta) = \E\biggl[ \sum_{s = 1}^{S_i} \Psi_s(\bm\delta) \biggr] = \E\biggl[ \sum_{s = 1}^{s_{\max}} \mathbbm{1}\{S_i \geq s\} \Psi_s(\bm\delta) \biggr] = \sum_{s = 1}^{s_{\max}} \P(S_i \geq s) \Psi_s(\bm\delta).
\end{align*}
Now, from example 1 of \cite{hines_demystifying_2022}, the influence function for $\P(S_i \geq s)$ is given by
\begin{align*}
    \psi^{(S_i \geq s)}(\bm\delta) = \mathbbm{1}\{S_i \geq s\} - \P(S_i \geq s).
\end{align*}
Therefore, by product rule, the influence function for $\bar{\Psi}(\bm\delta)$ is then given by
\begin{align*}
    \bar{\psi}(\cD_{i}; \bm\delta, \bm\eta, \bar\Psi) &= \Psi_s(\bm\delta) \psi^{(S_i \geq s)}(\bm\delta)  + \mathbbm{1}\{S_{is} \geq s\}
\psi_s(\cD_{is}; \bm\delta, \bm\eta_s, \Psi_s)\\
&= \sum_{s=1}^{s_{\max}}
\biggl\{
\{ \mathbbm{1}\{S_i \geq s\} - \P(S_i \geq s) \}\Psi_s(\bm\delta)
+
\mathbbm{1}\{S_i \geq s\}
\psi_s(\cD_{is}; \bm\delta, \bm\eta_s, \Psi_s)
\biggr\}\\
&= \sum_{s=1}^{s_{\max}} \varphi_s(\cD_{is}; \bm\delta, \bm\eta_s, \Psi_s) - \sum_{s=1}^{s_{\max}} \P(S_i \geq s) \Psi_s(\bm\delta)\\
&= \sum_{s=1}^{s_{\max}} \varphi_s(\cD_{is}; \bm\delta, \bm\eta_s, \Psi_s) - \bar{\Psi}(\bm\delta).
\end{align*}
\end{proof}
\end{comment}

\end{document}